\documentclass[a4paper,cleveref,british]{lipics-v2021}

\pdfoutput=1

\usepackage[utf8]{inputenc}
\usepackage[T1]{fontenc}
\usepackage{pdflscape}
\usepackage{lmodern}

\usepackage{xcolor}
\definecolor{amazing}{RGB}{254,67,101}
\definecolor{cardinal}{HTML}{BB333C}
\definecolor{niceblack}{HTML}{020300}

\clearpage{}
\usepackage{bbm}
\usepackage{mathrsfs}
\usepackage{latexsym}
\usepackage{microtype}

\usepackage{suffix}

\usepackage{xpunctuate}
\usepackage[all,british]{foreign} \redefnotforeign[ie]{i.e\xperiodafter} \redefnotforeign[eg]{e.g\xperiodafter}

\usepackage{xkvltxp}
\usepackage[draft,english,silent]{fixme}
\usepackage{todonotes}

\usepackage{xfrac} \usepackage[style=english,english=british]{csquotes}

\usepackage{amsmath, amssymb, amsfonts} \usepackage{amsthm}

\usepackage{mathtools}

\usepackage[full,small]{complexity}

\usepackage[final,notref,color]{showkeys} 

\usepackage{xargs}
\usepackage{afterpage} 

\PassOptionsToPackage{dvipsnames,usenames,table}{xcolor}
\usepackage{tikz}
\usetikzlibrary{calc}

\usepackage{booktabs} \usepackage{longtable}
\usepackage{float}
\usepackage{wrapfig}
\usepackage{floatflt}
\usepackage{framed}
\usepackage{dcolumn}
\usepackage{adjustbox}

\usepackage{xspace}
\usepackage{xifthen}

\usepackage{url}
\usepackage{scrtime}

\usepackage[longend,vlined]{algorithm2e}

\usepackage{marginnote}

\makeatletter
\newcommand{\raisemath}[1]{\mathpalette{\raisem@th{#1}}}
\newcommand{\raisem@th}[3]{\raisebox{#1}{$#2#3$}}
\makeatother

\def\N{\mathbb{N}} 	
\def\Z{\mathbb{Z}}

\newclass{\paraNP}{paraNP}

\newcommand\overbar[2][3]{{}\mkern#1mu\overline{\mkern-#1mu#2}}

\newcommand\restr[2]{{\left.\kern-\nulldelimiterspace #1 \vphantom{\big|} \right|_{#2} }}

   \def\sminor^#1{
    \preccurlyeq_{\mathrlap{\mathsf{m}}}^{#1}
} \def\ssminor^#1{\mathbin{\dot\preccurlyeq_{{\mathrlap{\mathsf{m}}}}^{#1}}\,
} \def\stminor^#1{
    \preccurlyeq_{\mathrlap{\mathsf{t}}}^{#1}
} \def\sstminor^#1{\mathbin{\dot\preccurlyeq_{{\mathrlap{\mathsf{t}}}}^{#1}}\,
}

\def\grad_#1{\nabla\!_{#1}}
\def\sgrad_#1{{\dot\nabla}\!_{#1}}
\def\topgrad_#1{\widetilde \nabla\!_{#1}}
\def\stopgrad_#1{\dot{\widetilde\nabla}\!_{#1}}
\def\topomega_#1{\widetilde \omega_{#1}}

\def\colnum_#1{ \operatorname{col}_{#1} }
\def\wcolnum_#1{ \operatorname{wcol}_{#1} }
\def\adm_#1{ \operatorname{adm}_{#1} }

\renewcommand{\leq}{\leqslant}

\renewcommand{\geq}{\geqslant}

\renewcommand{\epsilon}{\varepsilon}
\newcommand{\eps}{\epsilon}

\renewcommand{\emptyset}{\varnothing}

\newlength{\convarrowwidth}
\usepackage{stmaryrd}

\newcommand{\widthm}[1]{ \mathbf{#1} }

 \DeclareMathOperator{\width}{ \widthm{width} } 

\def\YYYY{{Y_0 \uplus Y_1 \uplus \cdots \uplus Y_\ell}} \WithSuffix\def\YYYY'{{Y'_0 \uplus Y'_1 \uplus \cdots \uplus Y'_{\ell'}}}

\def\vc{\operatorname{\mathbf{vc}}}

\usepackage{pifont}\newcommand{\yaay}{\kern4pt \ding{51} \kern-8pt \ding{51}}

\def\any{\mathord{\color{black!33}\bullet}}

\DeclarePairedDelimiter\ceil{\lceil}{\rceil}
\DeclarePairedDelimiter\floor{\lfloor}{\rfloor}

\newtheorem*{open*}{Open question}
\usepackage{thmtools, thm-restate}

\renewcommand*\etal{\xperiodafter{\emph{et~al}}}

\newcommand{\defineq}{\coloneqq}

\newcommand*\varrule[1][0.4pt]{\leavevmode\leaders\hrule height#1\hfill\kern0pt}

\newenvironment{tightcenter}
 {\parskip=0pt\par\nopagebreak\centering}
 {\par\noindent\ignorespacesafterend}

\newlength{\RoundedBoxWidth}
\newsavebox{\GrayRoundedBox}
\newenvironment{GrayBox}[1]{\setlength{\RoundedBoxWidth}{\textwidth-4.5ex}
    \def\boxheading{#1}
    \begin{lrbox}{\GrayRoundedBox}
       \begin{minipage}{\RoundedBoxWidth}}{\end{minipage}
    \end{lrbox}\begin{tightcenter}\begin{tikzpicture}\node(Text)[draw=black!20,fill=white,rounded corners,inner sep=2ex,text width=\RoundedBoxWidth]{\usebox{\GrayRoundedBox}};
        \coordinate(x) at (current bounding box.north west);
        \node [draw=white,rectangle,inner sep=3pt,anchor=north west,fill=white] 
        at ($(x)+(6pt,.75em)$) {\boxheading};
    \end{tikzpicture}
    \end{tightcenter}\vspace{0pt}\ignorespacesafterend
}

\newenvironment{ProblemDef*}[2][]{\noindent\ignorespaces \FrameSep=6pt\parindent=0pt\vspace*{-.5em}
                \ifthenelse{\isempty{#1}}{\begin{GrayBox}{\textsc{#2}}}{\begin{GrayBox}{\textsc{#2} parametrised by~{#1}}}
                \begin{tabular*}{\textwidth}{@{\hspace{.1em}} >{\itshape} p{1.6cm} p{0.8\textwidth} @{}}}{
                \end{tabular*}\end{GrayBox}\vspace*{-.5em}
                \ignorespacesafterend
            }       

\renewenvironmentx{leftbar}[2][1=2pt, 2=5pt]{\def\FrameCommand{\vrule width #1 \hspace{#2}}\MakeFramed {\advance\hsize-\width \FrameRestore}}{\endMakeFramed}

\newlength{\wleft}  \newlength{\wright}

\definecolor{Maroon}{cmyk}{0, 0.87, 0.68, 0.32}
\definecolor{RoyalBlue}{cmyk}{1, 0.50, 0, 0}
\definecolor{Black}{cmyk}{0, 0, 0, 0}
\definecolor{White}{rgb}{1, 1, 1}
\definecolor{Red}{rgb}{.9, 0, .1}
\definecolor{Magenta}{rgb}{.7, 0, .7}

\makeatletter
\let\orgdescriptionlabel\descriptionlabel
\def\@savelabel{}
\renewcommand*{\descriptionlabel}[1]{\let\orglabel\label
  \let\label\@gobble
  \phantomsection
  \def\@savelabel{#1}
  \edef\@currentlabel{{\def\hfil{}#1}}\edef\@currentlabelname{#1}\let\label\orglabel
  \orgdescriptionlabel{#1}}
\makeatother
\makeatletter
\def\namedlabel#1#2{\begingroup
   \def\@currentlabel{#1}\label{#2}\endgroup
}
\makeatother

\usepackage{hyphenat}
 \renewcommand{\th}{\ifmmode ^\mathrm{th}\else \textsuperscript{th}\xspace \fi }
\newcommand{\nd}{\ifmmode ^\mathrm{nd}\else \textsuperscript{nd}\xspace \fi }
\newcommand{\rd}{\ifmmode ^\mathrm{rd}\else \textsuperscript{rd}\xspace \fi }

\def\Mobius{M\"{o}bius\xspace}

\usepackage{environ}\NewEnviron{nofiitable}{\noindent\ignorespaces 
  \[\arraycolsep=1.4pt\begin{array}{rcr p{1cm} rcr}
    \BODY
  \end{array}
  \]
}

\newcolumntype{m}{>{$}l<{$}}
\newcolumntype{M}{>{$\displaystyle}l<{$}}

\newcolumntype{L}{l}
\newcolumntype{C}{c}
\newcolumntype{R}{r}
\newcolumntype{X}{>{\global\let\currentrowstyle\relax}}
\newcolumntype{^}{>{\currentrowstyle}}

\clearpage{}

\renewcommand{\G}{\mathbb G}
\newcommand{\X}{\mathbb X}
\newcommand{\Y}{\mathbb Y}
\renewcommand{\Z}{\mathbb Z}
\renewcommand{\H}{\mathbb H}

\newcommand{\ix}{\iota} \newcommand{\lc}{\operatorname{lc}}
\newcommand{\asym}{\operatorname{hoa}}


\authorrunning{P.\ G.\ Drange, P.\ Greaves, I.\ Muzi, and F.\ Reidl}


\nolinenumbers

\hideLIPIcs

\title{Computing complexity measures of degenerate graphs}

\author{Pål Grønås Drange}{University of Bergen, Norway}{Pal.Drange@uib.no}
{https://orcid.org/0000-0001-7228-6640}{Supported by the Research council of Norway, grant number~329745: \emph{Machine Teaching for Explainable AI}.}

\author{Patrick Greaves}{Birkbeck, University of London, UK}{p.greaves@bbk.ac.uk}{}{}

\author{Irene Muzi}{Birkbeck, University of London, UK}{i.muzi@bbk.ac.uk}{https://orcid.org/0000-0003-2410-6523}{}

\author{Felix Reidl}{Birkbeck, University of London, UK}{f.reidl@bbk.ac.uk}{https://orcid.org/0000-0002-2354-3003}{}

\newcommand{\DataStruct}{\ensuremath{Q}}

\begin{document}

\maketitle
\begin{abstract}
  We show that the VC-dimension of a graph can be computed in time
  $n^{\lceil\log d+1\rceil} d^{O(d)}$, where~$d$ is the degeneracy of the input graph. The core idea of our algorithm is a data structure to
  efficiently query the number of vertices that see a specific subset of
  vertices inside of a (small) query set. The construction of this data structure takes time $O(d2^dn)$, afterwards queries can be computed efficiently using fast Möbius inversion.

  This data structure turns out to be useful for a range of tasks,
  especially for finding bipartite patterns in degenerate graphs, and we outline an efficient algorithms for counting the number of times specific patterns occur in a graph. The largest factor in the running time of this algorithm is~$O(n^c)$, where~$c$ is a parameter of the pattern we call its \emph{left covering number}.

  Concrete applications of this algorithm include counting the number of (non-induced) bicliques in linear time, the number of co-matchings in quadratic time, as well as a constant-factor approximation of the ladder index in linear time.

  Finally, we supplement our theoretical results with several
  implementations and run experiments on more than 200 real-world
  datasets---the largest of which has 8 million edges---where we obtain
  interesting insights into the VC-dimension of real-world networks.
\end{abstract}

\section{Introduction}

Our work began with the simple question: What is the  Vapnik--Chervonenkis
(VC) dimension of real-world networks? That is, what is the largest vertex set $X$ such that every subset~$X' \subseteq X$ is the neighbourhood (when restricted to $X$) of some vertex in the network?

This parameter, developed in the
context of learning theory, happens to be extremely useful in the theory of
sparse graphs (\eg \cite{NComplexity17}) and is one possible method of capturing the
``complexity'' of an object. It is therefore a natural statistic to consider
when a) trying to categorise networks and b) identifying structural
properties that can be leveraged to design efficient algorithms.

As the best-known general algorithm to compute the VC-dimension of a graph takes~$O(n^{\log n})$ time, and in fact the problem being \LOGNP-complete~\cite{papadimitriou1996limited}, we investigated whether a better algorithm is possible if we assume our input graph to be sparse, more precisely, to be $d$-degenerate. This choice is motivated by the observation that the degeneracy for most real-world networks is small (see \eg \cite{demaine2019structural} and the results in the appendix).

Our first achievement is an algorithm that computes the VC-dimension of a $d$-degenerate graph in time~$O(n^{\ceil{\log d + 1}} d^{O(d)})$.  A core concept is a novel data structure which enables us to efficiently query the
size of the intersection of several neighbourhoods
for a small set of vertices, described in Section~\ref{sec:gener-patt-find}, which we use to quickly determine whether a given candidate set is shattered by its neighbours.

But the general idea of this algorithm can be generalised to other bipartite ``patterns'' like bicliques, co-matchings, and ladders (defined in Section~\ref{sec:patterns}). These objects are also closely related to notions of ``complexity'' of graphs. They appear, for example, in the study of graph width measures \cite{eiben2021unifying} and algorithm design for sparse classes \cite{ProgressiveAlgs19} (see also there for connections to stability theory). Our general pattern-finding algorithm presented in Section~\ref{sec:gener-patt-find} can count bicliques in linear time, co-matchings in quadratic time and find partial ladders in linear time, see Section~\ref{sec:concr-appl} for these and further results.

Dense structures like cliques or bicliques are famously important in the analysis of networks, and we suggest that co-matchings and ladders might be of similar interest---but without a program to compute them, we cannot hope for these statistics to be trialled in practice. We therefore implemented algorithms to compute the VC-dimension, ladder index, maximum biclique\footnote{There are probably faster programs to compute bicliques in practice, we  compute this statistic here as a baseline.} and maximum co-matching of a graph. To establish their practicality, we ran these four algorithms on 206 real-world networks from various sources, see Section~\ref{sec:experiments}. The VC-dimension algorithm in our experiments terminated within 10 minutes on networks with up to $\sim$33K vertices, the other three on networks up to $\sim$93K vertices. This is already squarely in the region of ``practical'' for certain types of networks and we believe that with further engineering---in particular to improve space efficiency---our implementation can be used to compute these statistics on much larger networks.

\medskip
\noindent \textbf{\textsf{Prior work.}}
We briefly mention a few relevant previous articles on the subject.
Eppstein, Löffler, and Strash~\cite{eppstein2013listing} gave an
algorithm for enumerating maximal cliques in $d$-degenerate graphs in
$O(dn3^{d/3})$ time, \ie, fixed-parameter tractable time when
parameterized by the degeneracy.  They also give experimental results
showing that their algorithm works well on large real-world networks.
Bera, Pashanasangi, and Seshadhri~\cite{bera2020linear}, extending the
classic result by Chiba and
Nishizeki~\cite{chiba1985arboricitysubgraph}, show that for all
patterns~$H$ of size less than six, we can count the number of
appearances of~$H$ in a $d$-degenerate graph~$G$ in time
$O(m \cdot d^{k-2})$, where~$m$ is the number of edges in~$G$ and~$k$ is
the number of vertices in~$H$.
Recently, Bressan, and Roth~\cite{bressan2022exact} gave algorithms for
counting copies of a graph~$H$ in a $d$-degenerate graph~$G$ in time
$f(d,k) \cdot n^{\mathbf{im}(H)} \log n$, for some function~$f$,
where~$k$ again is the number of vertices in~$H$, $n$ the number of
vertices in~$G$, and $\mathbf{im}(H)$ is the size of a largest induced
matching in~$H$.

\section{Preliminaries}
\label{sec:preliminaries}

For an integer $k$, we use $[k]$ as a short-hand for the set $\{0, 1, 2, \dots, k-1\}$.
We use blackboard bold letters like~$\X$ to denote sets~$X$
 associated with a total order~$<_\X$. The
 \emph{index function}~$\ix_\X \colon X \to \N$ maps elements of~$X$ to their
 corresponding position in~$\X$. We extend this function to sets via
 $\ix_\X(S) = \{ \ix_\X(s) \mid s \in S \}$. For any integer~$i \in [|X|]$ we write
 $\X[i]$ to mean the $i$th element in the ordered set.
 An \emph{index set}~$I$ for~$\X$ is simply a subset of~$[|X|]$ and we extend the index notation to sets via~$\X[I] \defineq \{ \X[i] \mid i \in I \}$.
We write $\pi(H)$ for the set of all permutations of~$H$.

For a graph $G$ we use $V(G)$ and $E(G)$ to refer to its vertex- and edge-set,
respectively. We used the short hands $|G| \defineq |V(G)|$ and $\|G\| \defineq |E(G)|$.

An \emph{ordered graph} is a pair $\G = (G, <)$ where $G$ is a graph and $<$ a
total ordering of $V(G)$. We write $<_\G$ to denote the ordering for a given
ordered graph and extend this notation to the derived relations $\leq_\G$,
$>_\G$, $\geq_\G$.

We use the same notations for graphs and ordered graphs, additionally we write
$N^-(u) \defineq \{ v \in N(u) \mid v <_\G u \}$ for the \emph{left neighbourhood}
and $N^+(u) \defineq \{ v \in N(u) \mid v >_\G u \}$ for the \emph{right
neighbourhood} of a vertex $u \in \G$. We further use $d_\G^-(u)$ and
$d_G^+(u)$ for the left and right degree, as well as $\Delta^-(\G) \defineq \max_{u
\in \G} d_\G^-(u)$ and $\Delta^+(\G) \defineq \max_{u \in \G} d_\G^+(u)$. We omit the graphs in the subscripts if clear from the context.

A graph~$G$ is $d$-degenerate if there exists an ordering $\G$ such that
$\Delta^-(\G) \leq d$. An equivalent definition is that a graph is
$d$-degenerate if every subgraph has a vertex of degree at most $d$. The
number of edges in a $d$-degenerate graph is bounded by~$dn$ and many
important sparse graph classes---bounded treewidth, planar graphs, graphs
excluding a minor---have finite degeneracy. The degeneracy ordering of a graph can be computed in time $O(n + m)$ \cite{matula_smallest-last_1983}, and $O(dn)$ for $d$-degenerate graphs.

Let $\mathcal F \subseteq 2^U$ be a set family over~$U$. We define the intersection of a set family with set~$X \subseteq U$ as $\mathcal F \cap X \defineq \{ F \cap X \mid F \in \mathcal F \}$. A set~$X \subseteq U$ is then \emph{shattered} by $\mathcal F$
if $\mathcal F \cap X = 2^X$. The \emph{graph representation} of a set family~$\mathcal F$ is the bipartite graph~$G(\mathcal F) = (\mathcal F, U, E)$ where
for each $F \in \mathcal F$ and $x \in U$ we have the edge $Fx \in E$ iff $x \in F$. In the other direction, we define for a graph~$G$ its \emph{neighbourhood set system} $\mathcal F(G) \defineq \{ N(v) \mid v \in G\}$.

The Vapnik--Chervonenkis dimension
(VC-dimension) of a set family~$\mathcal F \subseteq 2^U$ is the size of the
largest set in~$U$ that is shattered by~$\mathcal F$ and we write this
quantity as~$\vc(\mathcal F)$. The VC-dimension of a graph~$G$ is defined as
the VC-dimension of its neighbourhood set system, \ie
$\vc(G) \defineq \vc(\mathcal F(G))$.

\subsection{Set dictionaries}

\noindent
In the following we will make heavy use of data structures that model
functions of the form $f\colon 2^U \to \Z$ for some universe $U$. Since the
arguments in our use-case are assumed to be small, we use prefix-tries~\cite{sedgewick2011algorithms} in our theoretical analysis (see notes on practical implementations below):\looseness-1

\begin{definition}[Subset dictionary]
  Let~$U$ be a set and let~$\mathbb U$ be an arbitrary total order of~$U$.
  A \emph{subset dictionary}~$D$ over~$U$ associates a key $X \subseteq U$ with an integer $D[X]$ by storing the sequence $\mathbb X$ of~$X$ under~$<_{\mathbb U}$ in a prefix trie.
\end{definition}

\noindent
Accordingly, insertion/update/deletion of a value for a key~$X$ takes time~$O(|X|)$ \emph{if} we can assume the key~$X$ to be present in some canonical order.  Our algorithms all work on graphs imbued with a (degeneracy) ordering and we will sort the left-neighbourhood $N^-(\any)$ of each vertex according to this global ordering, which we will simply call ``sorting the left-neighbourhoods'' for brevity. Subsets of these left-neighbourhoods are assumed to inherit this ordering, which covers all operations that we will need in our algorithms, which in conclusion means that we can assume that all sets used as keys in subset dictionaries have a canonical ordering.

Unless otherwise noted, we will use the convention that~$D[X] = 0$ for all keys~$X$ that have not been inserted into~$D$.

\subsection{Bipartite patterns and left-covers}\label{sec:patterns}

\begin{definition}[Pattern]
  A \emph{pattern} $H$ is a complete graph whose edges are partitioned into
  sets $B$, $R$, and $W$ (black, red and white). We say that a graph $G$
  \emph{contains} $H$ (or $H$ \emph{appears in} $G$) if there exists a vertex set $S \subseteq V(G)$ and a
  bijection $\phi\colon V(H) \to S$ such that $uv \in B \implies
  \phi(u)\phi(v) \in E(G)$ and $uv \in R \implies \phi(u)\phi(v) \not \in
  E(G)$.

  We say that a pattern~$H$ is \emph{bipartite} if the vertex set of $H$ can
  be partitioned into two sets $X,Y$ such that all edges inside of $X$ and
  inside of $Y$ are white.
\end{definition}

\noindent
For a vertex $v \in V(H)$ we write $N(v)$ to denote its neighbours according to
the black edge relation only. An \emph{ordered} pattern~$\H$ is a pattern
whose vertex set comes with a linear order $<_\H$. Given a vertex $v \in
\H$, we write $N^-(u) \defineq \{ v \in N(u) \mid v <_\H u \}$.

A \emph{ladder} (sometimes called a chain graph) $L_n$ of size~$n$ is a bipartite pattern defined on two vertex sequences $A = (a_i)_{i \in [n]}$ and $B = (b_i)_{i \in [n]}$, where $a_i b_j \in B$ if~$i > j$ and
$a_i b_j \in R$ otherwise. Note that for any $1 \leq l \leq r \leq n$ the subgraph induced by the sequences $(a_i)_{i\in[l,r]}$ and $(b_i)_{i\in[l,r]}$ induces a ladder.
A \emph{semi-ladder}~$\tilde L_n$ has the same black edges, but only the edges $a_ib_i$, $i \in [n]$ are red.
All the remaining edges are white:

\tikzset{vertex/.style={circle,fill=niceblack!15,minimum size=18pt,inner sep=0pt
    },red edge/.style={draw,thick,-,cardinal!80
    },black edge/.style={draw,line width=.5pt,-,niceblack
    }
}

\begin{center}
\scalebox{.9}{\begin{tikzpicture}[scale=1.2]
    \node[vertex] (a1) at (1,0) {$a_1$};
    \node[vertex] (a2) at (2,0) {$a_2$};
    \node[vertex] (a3) at (3,0) {$a_3$};
    \node[vertex] (a4) at (4,0) {$a_4$};
    \node[vertex] (a5) at (5,0) {$a_5$};
    \node[vertex] (a6) at (6,0) {$a_6$};

    \node[vertex] (b1) at (1,-1.5) {$b_1$};
    \node[vertex] (b2) at (2,-1.5) {$b_2$};
    \node[vertex] (b3) at (3,-1.5) {$b_3$};
    \node[vertex] (b4) at (4,-1.5) {$b_4$};
    \node[vertex] (b5) at (5,-1.5) {$b_5$};
    \node[vertex] (b6) at (6,-1.5) {$b_6$};

    \draw[red edge] (a1) -- (b1);
    \draw[red edge] (a1) -- (b2);
    \draw[red edge] (a1) -- (b3);
    \draw[red edge] (a1) -- (b4);
    \draw[red edge] (a1) -- (b5);
    \draw[red edge] (a1) -- (b6);

    \draw[black edge] (a2) -- (b1);
    \draw[red edge] (a2) -- (b2);
    \draw[red edge] (a2) -- (b3);
    \draw[red edge] (a2) -- (b4);
    \draw[red edge] (a2) -- (b5);
    \draw[red edge] (a2) -- (b6);

    \draw[black edge] (a3) -- (b1);
    \draw[black edge] (a3) -- (b2);
    \draw[red edge] (a3) -- (b3);
    \draw[red edge] (a3) -- (b4);
    \draw[red edge] (a3) -- (b5);
    \draw[red edge] (a3) -- (b6);

    \draw[black edge] (a4) -- (b1);
    \draw[black edge] (a4) -- (b2);
    \draw[black edge] (a4) -- (b3);
    \draw[red edge] (a4) -- (b4);
    \draw[red edge] (a4) -- (b5);
    \draw[red edge] (a4) -- (b6);

    \draw[black edge] (a5) -- (b1);
    \draw[black edge] (a5) -- (b2);
    \draw[black edge] (a5) -- (b3);
    \draw[black edge] (a5) -- (b4);
    \draw[red edge] (a5) -- (b5);
    \draw[red edge] (a5) -- (b6);

    \draw[black edge] (a6) -- (b1);
    \draw[black edge] (a6) -- (b2);
    \draw[black edge] (a6) -- (b3);
    \draw[black edge] (a6) -- (b4);
    \draw[black edge] (a6) -- (b5);
    \draw[red edge] (a6) -- (b6);
\end{tikzpicture}}\hspace*{1em}
\scalebox{.9}{\begin{tikzpicture}[scale=1.2]
    \node[vertex] (a1) at (1,0) {$a_1$};
    \node[vertex] (a2) at (2,0) {$a_2$};
    \node[vertex] (a3) at (3,0) {$a_3$};
    \node[vertex] (a4) at (4,0) {$a_4$};
    \node[vertex] (a5) at (5,0) {$a_5$};
    \node[vertex] (a6) at (6,0) {$a_6$};

    \node[vertex] (b1) at (1,-1.5) {$b_1$};
    \node[vertex] (b2) at (2,-1.5) {$b_2$};
    \node[vertex] (b3) at (3,-1.5) {$b_3$};
    \node[vertex] (b4) at (4,-1.5) {$b_4$};
    \node[vertex] (b5) at (5,-1.5) {$b_5$};
    \node[vertex] (b6) at (6,-1.5) {$b_6$};

    \draw[red edge] (a1) -- (b1);

    \draw[black edge] (a2) -- (b1);
    \draw[red edge] (a2) -- (b2);

    \draw[black edge] (a3) -- (b1);
    \draw[black edge] (a3) -- (b2);
    \draw[red edge] (a3) -- (b3);

    \draw[black edge] (a4) -- (b1);
    \draw[black edge] (a4) -- (b2);
    \draw[black edge] (a4) -- (b3);
    \draw[red edge] (a4) -- (b4);

    \draw[black edge] (a5) -- (b1);
    \draw[black edge] (a5) -- (b2);
    \draw[black edge] (a5) -- (b3);
    \draw[black edge] (a5) -- (b4);
    \draw[red edge] (a5) -- (b5);

    \draw[black edge] (a6) -- (b1);
    \draw[black edge] (a6) -- (b2);
    \draw[black edge] (a6) -- (b3);
    \draw[black edge] (a6) -- (b4);
    \draw[black edge] (a6) -- (b5);
    \draw[red edge] (a6) -- (b6);
\end{tikzpicture}}
\end{center}

\noindent
The \emph{Ladder index} of a graph~$G$ is the largest~$n$ such that $G$ contains the pattern $L_n$.

A \emph{co-matching} $\overbar M_n$ (also called \emph{crown}) has black edges $a_ib_j$ for~$i \neq j$ and red edges $a_ib_i$ for~$i \in [n]$.

\begin{center}
\scalebox{.9}{\begin{tikzpicture}[scale=1.3]
    \node[vertex] (a1) at (1,0) {$a_1$};
    \node[vertex] (a2) at (2,0) {$a_2$};
    \node[vertex] (a3) at (3,0) {$a_3$};
    \node[vertex] (a4) at (4,0) {$a_4$};
    \node[vertex] (a5) at (5,0) {$a_5$};
    \node[vertex] (a6) at (6,0) {$a_6$};

    \node[vertex] (b1) at (1,-1.5) {$b_1$};
    \node[vertex] (b2) at (2,-1.5) {$b_2$};
    \node[vertex] (b3) at (3,-1.5) {$b_3$};
    \node[vertex] (b4) at (4,-1.5) {$b_4$};
    \node[vertex] (b5) at (5,-1.5) {$b_5$};
    \node[vertex] (b6) at (6,-1.5) {$b_6$};

    \draw[red edge] (a1) -- (b1);
    \draw[black edge] (a1) -- (b2);
    \draw[black edge] (a1) -- (b3);
    \draw[black edge] (a1) -- (b4);
    \draw[black edge] (a1) -- (b5);
    \draw[black edge] (a1) -- (b6);

    \draw[black edge] (a2) -- (b1);
    \draw[red edge] (a2) -- (b2);
    \draw[black edge] (a2) -- (b3);
    \draw[black edge] (a2) -- (b4);
    \draw[black edge] (a2) -- (b5);
    \draw[black edge] (a2) -- (b6);

    \draw[black edge] (a3) -- (b1);
    \draw[black edge] (a3) -- (b2);
    \draw[red edge] (a3) -- (b3);
    \draw[black edge] (a3) -- (b4);
    \draw[black edge] (a3) -- (b5);
    \draw[black edge] (a3) -- (b6);

    \draw[black edge] (a4) -- (b1);
    \draw[black edge] (a4) -- (b2);
    \draw[black edge] (a4) -- (b3);
    \draw[red edge] (a4) -- (b4);
    \draw[black edge] (a4) -- (b5);
    \draw[black edge] (a4) -- (b6);

    \draw[black edge] (a5) -- (b1);
    \draw[black edge] (a5) -- (b2);
    \draw[black edge] (a5) -- (b3);
    \draw[black edge] (a5) -- (b4);
    \draw[red edge] (a5) -- (b5);
    \draw[black edge] (a5) -- (b6);

    \draw[black edge] (a6) -- (b1);
    \draw[black edge] (a6) -- (b2);
    \draw[black edge] (a6) -- (b3);
    \draw[black edge] (a6) -- (b4);
    \draw[black edge] (a6) -- (b5);
    \draw[red edge] (a6) -- (b6);
\end{tikzpicture}}
\end{center}

\noindent
Finally, the \emph{shattered} pattern~$U_n$ of size~$n$ has a side~$S$ (the \emph{shattered} set) of size~$n$ and a side~$W$ (the \emph{witness} set) of size~$2^n$.
We index the vertices of~$W$ by subsets~$I \subseteq S$, then the vertex~$w_I$ has black edges into~$I$ and red edges into~$S \setminus I$:

\begin{center}
\scalebox{.9}{\begin{tikzpicture}[scale=1.3]
    \node[vertex] (a1) at (3.5,0) {$s_1$};
    \node[vertex] (a2) at (4.5,0) {$s_2$};
    \node[vertex] (a3) at (5.5,0) {$s_3$};

    \node[vertex] (b0) at (1,-1.5) {$w_\emptyset$};
    \node[vertex] (b1) at (2,-1.5) {$w_1$};
    \node[vertex] (b2) at (3,-1.5) {$w_2$};
    \node[vertex] (b3) at (4,-1.5) {$w_3$};
    \node[vertex] (b12) at (5,-1.5) {$w_{12}$};
    \node[vertex] (b13) at (6,-1.5) {$w_{13}$};
    \node[vertex] (b23) at (7,-1.5) {$w_{23}$};
    \node[vertex] (b123) at (8,-1.5) {$w_{123}$};

    \draw[red edge] (b0) -- (a1);
    \draw[red edge] (b0) -- (a2);
    \draw[red edge] (b0) -- (a3);

    \draw[black edge] (b1) -- (a1);
    \draw[red edge] (b1) -- (a2);
    \draw[red edge] (b1) -- (a3);

    \draw[red edge] (b2) -- (a1);
    \draw[black edge] (b2) -- (a2);
    \draw[red edge] (b2) -- (a3);

    \draw[red edge] (b3) -- (a1);
    \draw[red edge] (b3) -- (a2);
    \draw[black edge] (b3) -- (a3);

    \draw[black edge] (b12) -- (a1);
    \draw[black edge] (b12) -- (a2);
    \draw[red edge] (b12) -- (a3);

    \draw[black edge] (b13) -- (a1);
    \draw[red edge] (b13) -- (a2);
    \draw[black edge] (b13) -- (a3);

    \draw[red edge] (b23) -- (a1);
    \draw[black edge] (b23) -- (a2);
    \draw[black edge] (b23) -- (a3);

    \draw[black edge] (b123) -- (a1);
    \draw[black edge] (b123) -- (a2);
    \draw[black edge] (b123) -- (a3);
\end{tikzpicture}}
\end{center}

\begin{definition}[Left-cover, left-covering number]
  Given an ordered bipartite pattern $\mathbb H$ with bipartition $(X,Y)$, a
  \emph{left-cover} is a set of vertices $C \subseteq V(\mathbb H)$ such that
  either $X \subseteq N^-(C) \cup C$ or $Y \subseteq N^-(C) \cup C$. The
  \emph{left-covering number} $\lc(\mathbb H)$ is the minimum size of a left
  cover of~$\mathbb H$.
  For an (unordered) pattern~$H$ we define its left-covering number as
  \[
    \lc(H) \defineq \max_{\H \in \pi(H)} \lc(\H).
  \]
\end{definition}

\noindent
Note that we include the covering set~$C$ itself in the cover, this is
necessary since for a given ordering of a pattern some vertices might not have
right neighbours and can therefore not be covered by left neighbourhoods.

The left-covering number of a pattern is the first important measure that
will influence the running time of the main algorithm presented later. The
second important measure relates to the number of non-isomorphic ``half''-ordered
patterns we can obtain from a bipartite pattern, that is, how many distinct
objects we find by ordering one partition. A useful tool to concretise this
notion is the following function:

\begin{definition}[Signature]
  Let~$H$ be a bipartite pattern with bipartition~$(X,Y)$ and let~$\Z$
  be an ordering of~$Z \in \{X,Y\}$. Then the \emph{signature}~$\sigma_\Z(H)$ is
  defined as the multiset
  \[
    \sigma_\Z(H) \defineq \{\!\{ \ix_\Z(N(u)) \mid u \in (X \cup Y) \setminus Z \}\!\}.
  \]
  For orderings~$\Z, \Z' \in \pi(Z)$ we define the equivalence relation
  \[
    \Z \sim_H \Z' \iff \sigma_\Z(H) = \sigma_{\Z'}(H).
  \]
\end{definition}

\begin{definition}[Half-ordering asymmetry]
  Given a bipartite pattern~$H$ with bipartition~$(X,Y)$ and a partite
  set~$Z \in \{X,Y\}$, we define the \emph{half-ordering asymmetry}~$\asym(H,Z)$
  as the number of equivalence classes under the $\sim_H$ relation
  \[
    \asym(H,Z) \defineq \left| \pi(Z) / {}\sim_H \right|.
  \]
  We further define the half-ordering asymmetry of~$H$ as
  \[
    \asym(H) \defineq \max\{ \asym(H,X), \asym(H,Y) \}.
  \]
\end{definition}

\noindent
Alternatively, $\asym(H,Z) \defineq |\{\sigma_{\Z}(H) \mid \Z \in \pi(Z) \}|$.

\section{A general pattern-finding algorithm}
\label{sec:gener-patt-find}

\noindent
We first describe a general-purpose algorithm for finding patterns in
degenerate graphs. Afterwards, we will describe more specialised algorithms
using similar ideas to find specific patterns.

\begin{theorem}\label{thm:general-algorithm}
  Let $G$ be a $d$-degenerate graph and let $H$ be a bipartite pattern with
  bipartition $(X,Y)$ where~$|X| \geq |Y|$. Then after a preprocessing time
  of $O(|X|^{\lc(H)}|H|! + d2^dn)$, we can in time
  $
    O\big(n^{\lc(H)} (4d\lc(H))^{|X|} d|X|^3 \asym(H) \big)
  $
  count how often~$H$ appears in $G$.
\end{theorem}

\noindent
The main ingredient of our algorithm will be the following data structure:

\begin{theorem}\label{thm:data-structure}
  Let $\G$ be an ordered graph on~$n$ vertices with degeneracy~$d$. After a
  preprocessing time of~$O(d2^d n)$, we can, for any given~$S \subseteq V(G)$,
  compute a subset dictionary~$\DataStruct_S$ in time~$O(|S| 2^{|S|} + d|S|^2)$ which
  for any $X \subseteq S \subseteq V(G)$ answers the query
  \[
    \DataStruct_S[X] \defineq \big| \{ v \in G \mid S \cap N(v) = X \} \big|
  \]
  in time~$O(|X|)$.
\end{theorem}

\begin{lemma}\label{lemma:R}
  Let $\G$ be an ordered graph with degeneracy $d$. Then in time $O(d2^d n)$
  we can compute a subset dictionary $R$ over~$V(G)$ which for any $X \subseteq
  V(G)$ answers the query
  \[
    R[X] \defineq \big| \{ v \in G \mid X \subseteq N^-(v) \} \big|
  \]
  in time~$O(|X|)$.
\end{lemma}
\begin{proof}
  Given $\G$ as input, we compute $R$ as follows:

  \begin{algorithm}[H]
    Initialize $R$ as an empty trie storing integers\;

    \For{$u \in \mathbb G$}{
      \For{$X \subseteq N^-(u)$}{
        $R[X] \leftarrow R[X] + 1$ \quad // Non-existing keys are treated as zero
      }
    }
    \Return $R$\;
  \end{algorithm}

  \noindent
  Note that every update of the data structure with key~$X$ takes time
  $O(|X|)$, since $|X| \leq d$ it follows that the total initialisation time
  is bounded by $O(d2^d n)$.
\end{proof}

\begin{lemma}\label{lemma:count}
  Let $\G$ be an ordered graph with degeneracy $d$ and let $S \subseteq V(G)$. If we assume the subset dictionary $R$ of \Cref{lemma:R} is given, we can
  construct in time $O(|S| 2^{|S|} + d|S|^2)$ a subset dictionary
  $\DataStruct_S$ over~$S$ which for $X \subseteq S$ answer the query
  \[
    \DataStruct_S[X] \defineq \big|\{v \in G \mid S \cap N(v) = X\}\big|
  \]
  in time~$O(|X|)$.
\end{lemma}
\begin{proof}
  We first construct an auxiliary subset dictionary $\hat \DataStruct$ which
  for $X \subseteq S$ answers the query
  \[
    \hat \DataStruct_S[X] \defineq \big|\{v \in G \mid S \cap N^-(v) = X\}\big|
  \]
  in time~$O(|X|)$. We first prove the following claim which implies that
  $\hat \DataStruct_S$ is the (upwards) \Mobius inversion of $R$ over~$S$ and hence can
  be computed in time $O(|S| 2^{|S|})$ using Yate's algorithm~\cite{yates_design_1937,knuth_art_2014,bjorklund_set_2009}.
  \begin{claim}
  $
    \big|\{v \in G \mid S \cap N^-(v) = X\}\big| = \sum_{X \subseteq Y \subseteq S}  (-1)^{|Y\setminus X|} R[Y],
  $
  \end{claim}
  \begin{proof}
  First consider $v \not \geq_\G X$. Then $X$ cannot be contained in $N^-(v)$
  and therefore $v$ does not contribute to the left-hand side. Note that $v$
  is not counted by $R[Y]$ for any $Y \supseteq X$, therefore $v$ does not
  contribute to the right-hand side.

  Consider therefore $v \geq_\mathbb G X$. First, assume that $S \cap N^-(v) =
  X$ and therefore $v$ contributes to the left-hand side. Then $v$ is counted
  on the right-hand side exactly once by the term $R[X]$ which has a positive
  sign.

  Consider now $v$ with $S \cap N^-(v) \neq X$. If $X \not \subseteq N^-(v)$,
  then $v$ does not contribute to the left-hands side and it is not counted by
  any term $R[Y]$, $Y \supseteq X$ on the right-hand side. We are therefore
  left with vertices $v$ where $I \defineq S \cap N^-(v)$ satisfies $X \subset I$.
  Note that $I$ is counted by every term $R[Y]$ with $X \subseteq Y \subseteq
  I$. Since
  \[
    \sum_{X \subseteq Y \subseteq I} (-1)^{(Y\setminus X)} = \sum_{0 \leq k \leq |I \setminus X|} (-1)^k {|I\setminus X| \choose k} = 0
  \]
  we conclude that these counts of~$v$ cancel out and contribute a sum-total
  of zero to the right-hand side. This covers all cases and we conclude that
  the claim holds.
  \end{proof}

  \noindent
  It remains to be shown how the query $\DataStruct_S[X]$ can be computed using $\hat
  \DataStruct_S[X]$. To this end, consider a vertex~$v \in G$ where $S \cap N(v) \neq S
  \cap N^-(v)$ as these contribute to $\hat \DataStruct_S[X]$ but must not be counted by
  $\DataStruct_S[X]$. Note that any such vertex must be contained in $N^-(S)$ since $v$
  has at least one right-neighbour in~$S$. Accordingly, we apply the following
  correction to $\hat \DataStruct_S[X]$: \looseness-1

  \begin{algorithm}[H]
    \DontPrintSemicolon
    Let $\DataStruct_S = \hat \DataStruct_S$\;
    \For{$u \in N^-(S)$}{
      $\DataStruct_S[N^-(u) \cap S] \leftarrow \DataStruct_S[ N^-(u) \cap S] - 1$\;
      $\DataStruct_S[N(u) \cap S] \leftarrow \DataStruct_S[N(u) \cap S] + 1$\;
    }
  \end{algorithm}
  \noindent
  This correction takes time~$O(d|S|^2)$.
\end{proof}

\noindent
We are now ready to describe the pattern-counting algorithm.
\begin{proof}[Proof of \Cref{thm:general-algorithm}]
  The problem is trivial for~$|X| = 1$ since then the pattern is either a
  single edge or anti-edge. Thus assume~$|X| \geq 2$ in the following, in particular for the running time calculations.

  We first compute the left-covering number~$\lc(H)$ by simply brute-forcing
  all orderings of~$H$ in time~$O(|H|! \cdot \max\{|X|,|Y|\}^{\lc(H)}) = O
  (|H|! |X|^{\lc(H)})$. At the same time, whenever we find that a specific
  ordering~$\H$ has a minimal left-covering of~$X$, then we add the
  signature~$\sigma_\X(H)$ with $\X \defineq \H[X]$ to a collection~$\mathcal X$.
  Similarly, if we find that a minimal left-covering in~$\H$ covers~$Y$ we
  add the signature~$\sigma_\Y(H)$ with~$\Y \defineq \H[Y]$ to a
  collection~$\mathcal Y$. We will later use that~$|\mathcal X| \leq \asym(H,X)$
  and~$|\mathcal Y| \leq \asym(H,Y)$.

  We now compute an ordering~$\G$ for~$G$ of degeneracy~$d$ in time~$O(d n)$, sort the left-neighbourhoods in time $O(d \log d \cdot n)$ time and
  compute the data structure~$R$ as per \Cref{lemma:R} in time~$O
  (d2^d n)$. If we want to compute the number of times~$H$ appears in~$G$, we
  further need to initialise a subset dictionary~$K$.

  We now iterate through all subsets~$C \subseteq V(G)$ of size~$\lc(H)$ and
  for each such set we iterate through all subsets~$Z \subseteq N_\G^-
  (C) \cup C$ of size~$|X|$ or~$|Y|$,
  in total this takes time~$O(n^{\lc(H)} ((d+1)\lc(H))^{|X|})$. We describe the remainder of the algorithm for a set~$X =
  Z$ of size~$|X|$, the procedure for a set~$Y$ works analogously. Let $\X$ be the ordering of~$X$ in~$\G$.

  To verify that~$X$ can be completed into a pattern~$H$ in~$G$, we compute the data structure~$\DataStruct_{X}$ in time~$O(|X|2^{|X|} + d|X|^2)$ as per
  \Cref {thm:data-structure}. To check whether~$H$ exists
  in~$G$, we iterate through all signatures~$\sigma \in \mathcal X$
  and test whether~$\DataStruct_{X}[\X[A]] > 0$ for all index sets~$A \in \sigma$, this takes time~$O(|\mathcal X||X||Y|)$, in total the verification step for~$X$ takes time
  \[
    O\big((|X| 2^{|X|} + d|X|^2) \cdot |\mathcal X| |X| |Y| \big)
    = O\big(d|X|^3 2^{|X|} \asym(H) \big)
  \]
  where we used that~$|X| \geq |Y|$ and~$|X| \geq 2$. This bound also holds for checking~$Y$ since~$|\mathcal X| + |\mathcal Y| \leq 2\asym(H)$.
  Finally, if we exhaust all orderings of~$H$ without finding the pattern, we report that it does not exist in~$G$.

  To \emph{count} in how many ways~$X$ can be extended into the pattern~$H$ in~$G$, we compute
  \[
    c_{H,X}\defineq \sum_{\sigma \in \mathcal X} \prod_{A^{(k)} \in \sigma} {\DataStruct_{X}[\X[A]] \choose k}
  \]
  where~$k$ denotes the multiplicity of~$A$ in the multiset~$\sigma$.
  Note, however, that we have to take care not to double-count the contribution of~$X$ to the overall count as we might encounter the set~$X$ multiple times. To that end, we record the intermediate result by setting~$K[X] \defineq c_{H, X}$ and we forgo the above computation if~$X$ exists already as
  a key in~$K$. The computation of~$c_{H, X}$ and this additional book keeping takes time~$O(|X| + |\mathcal X||X||Y|)$, in total we arrive at the same
  running time $O\big(d|X|^3 2^{|X|} \asym(H) \big)$ like for the decision variant.
  After exhausting all orderings of~$H$ we report back the number of times~$H$ appears in~$G$ as the sum of all entries of~$K$.

  The total running time  of either variant of the algorithm is, as claimed,
  \begin{align*}
    & O\Big(|X|^{\lc(H)}|H|! + d2^dn + dn + n^{\lc(H)}\big((d+1)\lc(H)\big)^{|X|} \cdot  d|X|^3 2^{|X|} \asym(H) \Big) \\
    &= O\Big(|X|^{\lc(H)}|H|! + d2^dn + n^{\lc(H)} (4d\lc(H))^{|X|} d|X|^3 \asym(H) \Big). \qedhere
  \end{align*}
\end{proof}

\section{Concrete applications}
\label{sec:concr-appl}

\subsection{Finding bicliques and co-matchings}

We note that~$\lc(K_{t,t}) = 1$ and $\asym(K_{t,t}) = 1$, therefore the application of \Cref{thm:general-algorithm} gives the following:

\begin{corollary}
  Let~$G$ be a $d$-degenerate graph. Then we can compute the number of
  biclique patterns~$K_{s,t}$ ($s \geq t$) in time~$O\big(s \cdot (2s)! + d2^dn + n (4d)^s ds^3 \big)$.
\end{corollary}

\noindent
Let~$\overbar M_t$ be a co-matching on $2t$ vertices. We will assume in the following that the partite sets of~$\overbar M_t$ are~$X \defineq (x_1,\ldots,x_t)$ and $Y \defineq (y_1,\ldots,y_t)$ so that the edges~$x_iy_i$ for~$i \in [t]$ are forbidden.

\begin{lemma}
    $\lc(\overbar M_t) = 2$ and $\asym(\overbar M_t) = 1$.
\end{lemma}
\begin{proof}
  Let~$\bar {\mathbb M}_t$ be an ordering of~$\overbar M_t$ and let~$z$ be the last
  vertex in that order. Then~$N^-(z)$ covers all vertices of one partite set except one vertex~$z'$. Thus $\{z,z'\}$ is a left-cover of~$\bar {\mathbb M}_t$.

  To determine the half-ordering asymmetry, note that for \emph{every}
  ordering $\Z$ of~$Z \in \{ X,Y \}$ the signature~$\sigma_\Z(\overbar M_t)$
  is simply the set ${[t] \choose t-1}$, so the total number of signatures is one.
\end{proof}

\begin{corollary}
  Let~$G$ be a $d$-degenerate graph. Then we can compute the number of
  co-matching patterns~$\overbar M_t$ in time~$O\big(t^2 (2t)! + d2^dn + n^2 (8d)^t dt^3 \big)$.
\end{corollary}

\subsection{Finding shattered sets}

\noindent
A direct application of \Cref{thm:general-algorithm} to
locate a shattered pattern~$U_t$ is unsatisfactory as the
running time will include a factor of~$n^t$ since~$\lc(U_t) = t$. By the following observation, we can bound~$t$ by the degeneracy of the graph, but we can greatly improve the running time by further adjusting the algorithm.

\begin{observation}
  Let~$G$ be a $d$-degenerate graph. Then $\vc(G) \leq d+1$.
\end{observation}
\begin{proof}
  Assume~$S \subseteq V(G)$ is shattered by~$W \subseteq V(G)$, with $S = |\vc(G)|$. Let~$W'
  \subseteq W$ be those witnesses that have $|S|-1$ neighbours in~$S$.
  Then~$G[W' \cup S]$ induces a graph of  minimum degree~$|S|-1$ and we must
  have that $|S|-1 \leq d$ and accordingly $\vc(G) = |S| \leq d+1$.
\end{proof}

\noindent
The core observation that allows further improvements is that many orderings
of $U_{d+1}$ have degeneracy \emph{larger} than~$d$ and can therefore not appear in a $d$-degenerate graph. In particular, the ordering in which all witnesses of~$U_{d+1}$ appear before the shattered set has degeneracy~$2^{d+1}$ and can therefore be ruled out. We refine this idea further in the following lemma.

\begin{lemma}\label{lemma:shattered-small-lc}
  Let~$\G$ be a $d$-degenerate ordering of a graph~$G$. Let~$G$ contain the
  shattered pattern~$U_t$ and let~$\mathbb U_t \defineq \G[U_t]$ be its ordering.
  Then~$\lc(\mathbb U_t) \leq \ceil{\log d+1}$. Specifically, we either have that
  $t \leq \ceil{\log d + 1}$ or that~$\mathbb U_t$ can be covered by~$\ceil{\log d + 1}$
  witness vertices.
\end{lemma}
\begin{proof}
  Let $S = (s_1,\ldots,s_t)$ and $W = (w_1,\ldots,w_{2^t})$ be the vertices of
  $U_t$ in $G$ and let the indices of the variables reflect
  the ordering of the corresponding vertices in $\mathbb U_t$.

  Partition the set $S$ into $p \defineq \ceil{\log d + 1}$ sets $S_1,\ldots,S_{p}$ such that each set has size at least $\lfloor t/p \rfloor$ and at most $\lceil t/p \rceil$. For each set $S_i$ define the set of ``apex''-witnesses $A_i \defineq \{ w \in W \mid N(w) \supset S_i \}$. Note that, for all~$i \in [p]$,
  \[
    | A_i | = 2^{|S \setminus S_i|} \geq 2^{t-\lceil t/p \rceil}
    = 2^{\lceil t \frac{p-1}{p} \rceil}.
  \]
  We call a set $A_i$ \emph{good} if $\max_{\mathbb G} A_i > \max_{\mathbb G}
  S_i$, that is, at least one apex vertex from~$A_i$ can be found to the right
  of $S_i$. We now distinguish two cases:

  \smallskip\noindent
  \textbf{Case 1}. All $A_i$, $i \in [p]$, are good.

  \noindent
  It follows that $\mathbb U_t$ can be left-covered by taking one vertex from each~$A_i$, $i \in [p]$. We conclude that $\lc(\mathbb U_t) \leq p = \ceil{\log d + 1}$.

  \smallskip\noindent
  \textbf{Case 2}. Some $A_i$, $i \in [p]$, is not good.

  \noindent
  Let $u = \max_{\mathbb G} S_i$ be the last vertex in $S_i$, note that $A_i
  \leq_{\mathbb G} u$ and accordingly $A_i \subseteq N^-(u)$. But then we must
  have that $|A_i| \leq d$ and accordingly that
  \begin{align*}
      2^{\lceil t \frac{p-1}{p} \rceil} \leq d
      \iff&  \lceil t \frac{p-1}{p} \rceil \leq \log d
      \implies t \frac{p-1}{p}  \leq \log d
      \iff t \leq \frac{p}{p-1} \log d \\
      \iff& t \leq \frac{\ceil{\log d + 1}}{\ceil{\log d + 1}-1} \log d
      = \frac{\log d}{\ceil{\log d}} \ceil{\log d + 1} \leq \ceil{\log d + 1}.
  \end{align*}
We therefore find that $\lc(\mathbb U_t) \leq |S| \leq \ceil{\log d + 1}$.
\end{proof}

\begin{theorem}
  Let $G$ be a $d$-degenerate graph on~$n$ vertices. Then we can determine the
  VC-dimension of its neighbourhood set system~$\mathcal F(G)$ in time $O(n^{\ceil{\log d + 1}} d^{d+2} (2d \log d)^{d+1})$.
\end{theorem}
\begin{proof}
  We first compute an ordering~$\G$ of~$G$ with degeneracy~$d$ in time~$O(dn)$ and sort all left-neighbourhoods in time $O(d \log d \cdot n)$.
  Let~$p \defineq \ceil{\log d + 1}$ in the following.

  Let $\mathbb U_t = (S,W)$ be a shattered set of size $t \leq d+1$ in $\mathbb G$. By \Cref{lemma:shattered-small-lc} we then have that $\lc(\mathbb U_t) \leq p$. Therefore to locate the set~$S$ we first guess up to
  $p$ vertices and then exhaustively search through their (closed) left-neighbourhoods in time
  \[
    {n \choose p}{dp \choose t}
    \leq  \Big( \frac{en}{p} \Big)^p \Big( \frac{edp}{t} \Big)^t
    = O \left( n^{\ceil{\log d + 1}} (d \log d)^{d+1} \right)
    .
  \]
  Now that we can locate~$S$ we apply \Cref{thm:data-structure} in
  order to verify that~$S$ is indeed shattered: For each candidate set~$S$
  from the previous step, we compute a subset dictionary~$\DataStruct_S$ in time~$O(|S|
  2^{|S|} + d |S|^2) = O(d 2^d)$ and then check whether~$\DataStruct_S[X] > 0$ for
  each~$X \subseteq S$. This latter step takes time~$O(|S| 2^{|S|})$ and is
  therefore subsumed by the construction time of~$\DataStruct_S$. We conclude that the
  algorithm runs in total time
  \[
    O(d \log d \cdot n) + O(d2^d n) + O \left( n^{\ceil{\log d + 1}} (d \log d)^{d+1} \cdot d2^d \right)
    = O \left( n^{\ceil{\log d + 1}} d^{d+2} (2d \log d)^{d+1} \right)
  \]
  as claimed.
\end{proof}

\noindent
We note that the exponent of~$\lceil \log d + 1 \rceil$ in the running time is almost tight:

\begin{theorem}\label{thm:lower-bound}
  \textsc{Graph VC-dimension} parameterized by the degeneracy~$d$ of the
  input graph cannot be solved in time~$f(d) \cdot n^{o(\log d)}$
  unless all problems in \SNP{} can be solved in subexponential time.
\end{theorem}

\begin{proof}
  We adapt the W[1]-hardness reduction from \textsc{$k$-Clique} to
  \textsc{VC-dimension} by Downey, Evans, and Fellows~\cite{downey_parameterized_1993} and combine it
  with the result by Chen \etal~\cite{chen_strong_2006,chen_computational_2006}  which states that \textsc{$k$-Clique}
  cannot be solved in time $f(k) n^{o(k)}$ unless all problems in \SNP{} admit
  subexponential-time algorithms.

  Given an instance~$(H,k)$ for \textsc{$k$-Clique}, we construct a graph~$G$ as follows. We first create~$k$ copies $V_1,\ldots,V_k$ of~$V(H)$. For~$v \in H$, let us denote its copies by~$v^{(1)},\ldots,v^{(k)}$ with~$v^{(i)} \in V_i$ for~$i \in [k]$. We now add the following vertices and edges:
  \begin{itemize}
    \item A single isolated vertex~$w_0$,
    \item a vertex set~$W_1$ which contains one pendant vertex for each $v^{(i)}$, $v \in H$ and~$i \in [k]$,
    \item a vertex set~$W_2$ which for each edge~$uv \in H$ contains ${k \choose 2}$~vertices $w_{uv}^{ij}$, $i,j \in [k]$, each of which
    $u^{(i)}$ and~$v^{(j)}$ as its only neighbours, and
    \item a vertex set~$A$ which for each index set~$I \subseteq [k]$ contains a vertex~$a_I$ which is connected to all vertices in~$V_i$ for each $i \in I$.
  \end{itemize}
  Note that the graph is bipartite with partite sets~$\mathcal V \defineq V_1 \cup \cdots \cup V_k$ and $\mathcal W \defineq W_1 \cup W_2 \cup A$.

  Let us first show that if~$H$ contains a clique of size~$k$ then~$G$ contains a shattered set of size~$k$. Let~$u_1,\ldots,u_k$ be distinct vertices that form a complete graph in~$H$. We claim that then the set~$S \defineq \{u_1^{(1)},
  \ldots,u_k^{(k)}\}$ is shattered in~$G$. First, note that for every subset~$X \subset S$, $|X| \geq 3$, there exists a witness vertex~$a \in A$ such that
  $N(a) \cap S = X$. For the empty set we have the witness~$w_0$, for every singleton subset~$\{u\} \subseteq S$ we have that the pendant vertex~$p \in N(u) \cap W_1$ witnesses $\{u\}$. Therefore, only subsets of size exactly two need to be witnesses to shatter~$S$. Consider
  $\{u_i^{(i)}, u_j^{(j)}\} \subseteq S$ for~$i \neq j$. Since~$u_iu_j \in H$, the vertex~$w^{ij}_{u_i u_j}$ exists in~$W_2$ and its neighbourhood in~$S$
  is exactly~$\{u_i^{(i)}, u_j^{(j)}\}$. We conclude that all subsets of size two in~$S$ are witnessed as well and therefore~$S$ is shattered.

  In the other direction, assume that~$G$ contains a shattered set~$(S,W)$ of size~$k$. Without loss of generality, assume that~$k \geq 3$.
  \begin{claim}
    $S \subseteq \mathcal V$ and~$W \subseteq \mathcal W$.
  \end{claim}
  \begin{proof}
    Since~$G$ is bipartite we either have that $S \subseteq \mathcal V$ and~$W \subseteq \mathcal W$ or that $S \subseteq \mathcal W$ and~$W \subseteq \mathcal V$. Let us now show that the latter is impossible.

    Since~$k \geq 3$ we have that every vertex in~$S$ has degree at least four.
    Accordingly, $W$ cannot contain vertices from~$W_1$ or~$W_2$, which leaves us with~$W \subseteq A$. However, all vertices in~$V_i$, $i \in [k]$, have the exact same neighbours in~$A$. Therefore only~$k$ subsets of~$A$ are witnessed by vertices in~$\mathcal V$ and therefore the largest shattered set in~$A$ has size at most~$\log k$. We conclude that~$S$ cannot be contained in~$A$ and the claim holds.
  \end{proof}
  We now claim that $|S \cap V_i| = 1$ for all~$i \in [k]$. Assume otherwise,
  so let $u^{(i)}, v^{(i)} \in S$ for some~$i \in [k]$. But then the set
  $\{u^{i}, v^{i}\}$ cannot be witnessed: not by a vertex from~$W_1$, since it
  only contains vertices with one neighbour, not by a vertex from~$W_2$,
  since these vertices each have at most one neighbour in each set~$V_i$, and
  not by a vertex from~$A$ since we need all~$2^k - {k \choose 2} - k - 1$ vertices of~$A$ to witness subsets of~$S$ of size at least three.

  Therefore~$S$ intersects each~$V_i$ in exactly one vertex. Since~$S$ is shattered, every subset~$\{u^{(i)}, v^{(j)}\}$, $i \neq j$, is shattered. By the same logic as above, this can only be due to a witness~$w^{ij}_{uv} \in W_2$ and therefore $uv \in H$. We conclude that indeed~$u_1, \ldots, u_k$ induce a complete graph in~$H$, as claimed.

  Finally, we need to determine the degeneracy of~$G$. Consider the following elimination sequence: We first delete all of~$\{w_0\} \cup W_1 \cup W_2$, all of which have degree at most two. Note now that all vertices in~$\mathcal V$ have at most~$|A| < 2^k$ neighbours in~$A$, so we delete~$\mathcal V$ and then~$A$. In total, the maximum degree we encountered in this deletion sequence is~$< 2^k$.

  Assume we could solve \textsc{Graph VC-Dimension} in time~$f(d) n^{o
  (\log d)}$. In the above reduction the degeneracy of the constructed
  graph is $d < 2^k$, thus this running time for \textsc{Graph VC-Dimension} would imply a running time of
  \[
    f(d) \cdot n^{o(\log d)} = f(2^k) \cdot n^{o(\log 2^k)} = f(2^k) \cdot n^{o(k)}
  \]
  for \textsc{$k$-Clique}. We conclude that \textsc{VC-dimension} parameterized by the degeneracy of the input graph cannot be solved in time~$f(d) n^{o(\log d)}$ unless all problems in \SNP{} can be solved in subexponential time.
\end{proof}

\noindent
We note that \Cref{lemma:shattered-small-lc} allows us to approximate the VC-dimension of degenerate graphs.

\begin{theorem}
  Let $G$ be a $d$-degenerate graph on~$n$ vertices. Then for any $0 < \epsilon
  \leq 1$ we can approximate the VC-dimension of~$G$ in time~$O(d2^d (2n)^{\lceil \eps (1 + \log d) \rceil})$ within a factor of~$\epsilon$.
\end{theorem}
\begin{proof}
  We first compute a $d$-degenerate ordering $\G$ of~$G$ in time~$O(dn)$ and sort its left-neighbourhoods in time $O(d \log d \cdot n)$.
  Let~$U_t = (S,W)$ be the largest shattered set in~$G$ and let~$\mathbb U_t$ be its ordering in~$\G$. We further prepare the use of \Cref{thm:data-structure} by computing the necessary data structure in time~$O(d2^d n)$.

  Let $c \defineq \lceil \eps (1 + \log d) \rceil$.
  The algorithm now iterates over all~$C \subseteq V(G)$ of size~$c$ and searches the left-neighbourhood $L \defineq N^-[C]$ for a shattered set by first computing
  a subset dictionary $\DataStruct_L$ in time~$O(d2^d)$ and then finding the largest
  shattered subset $S \subseteq L$ by brute-force in time~$O(|L|2^{|L|}) = O(cd2^{cd})$.

  We claim that this simple algorithm computes the claimed approximation of the VC-dimension. By \Cref{lemma:shattered-small-lc} we either have that
  $t \leq \log d + 1$ or that~$\mathbb U_t$ can be left-covered by~$\log d + 1$
  witness vertices. In the first case, our algorithm will trivally locate an $\epsilon$-fraction of a maximal solution since it tests every set of size $c$.
  In the second case, the shattered set~$S$ of~$\mathbb U_t$ is covered by the left-neighbourhood of witness vertices
  $w_1,\ldots,w_p \in W$ for~$p \defineq \log d + 1$. Then by simple averaging, there exist $c$ witnesses~$W'$ such that~$|N^-[W'] \cap S| \geq c|S|/p = ct/(\log d + 1)$. Since the above algorithm will find the shattered set $N^-[W] \cap S$ when inspecting the left-neighbourhood of~$W$, we conclude that it will output at least a value of~$ct/(\log d + 1)$. In either case the approximation factor is~$\frac{c}{1+\log d} \geq \eps$, as claimed.
\end{proof}

\noindent
We would like to highlight the special case of $c=1$ of the above theorem as it provides us with a linear-time approximation of the VC-dimension, which is probably a good starting point for practical applications:

\begin{corollary}
  Let $G$ be a $d$-degenerate graph on~$n$ vertices. Then we can approximate the VC-dimension of~$G$ in time~$O(d 2^d n)$ within a factor of~$\frac{1}{1+\log d}$.
\end{corollary}

\subsection{Approximating the ladder and semi-ladder index}{}

Before we proceed, we note that degenerate graphs cannot contain arbitrarily long ladders:

\begin{observation}\label{obs:ladder-degen-bound}
  If $G$ is $d$-degenerate then $G$ cannot contain a ladder of length $2d+2$.
\end{observation}
\begin{proof}
  Note that a ladder of length~$t$ contains a complete bipartite
  graph $K_{\lfloor t/2 \rfloor, \lfloor t/2 \rfloor}$, \ie a subgraph of
  minimum degree $\lfloor t/2 \rfloor$. Therefore $t < 2d+2$.
\end{proof}

\noindent
Again we find that a direct application of \Cref{thm:general-algorithm} to ladder patterns does not yield a satisfying running time since~$\lc(L_t) \approx t/2$. However, we can always left-cover a large portion of a ladder with only one vertex:

\begin{observation}\label{obs:ladder-lc}
  Let $(A,B)$ induce a ladder of length $t$ in $G$. For every ordering
  $\mathbb G$ of $G$ there exists a vertex $u \in A \cup B$ such that $|N^-
  (u) \cap (A \cup B)| \geq \lfloor t/2 \rfloor$.
\end{observation}
\begin{proof}
  Let $A' \defineq (a_i)_{i \geq t/2}$ and $B' \defineq (b_i)_{i \leq t/2}$,  then $G
  [A'\cup B']$ contains a biclique with partite sets $A',B'$. Let $u \in
  A' \cup B'$ be the largest vertex according to $<_\mathbb G$, then $N^-
  (u) \cap (A', B')$ is either all of $A'$ or all of $B'$. In either case the
  claim holds.
\end{proof}

\begin{theorem}
  Let $G$ be a $d$-degenerate graph on~$n$ vertices and let~$t$ be its ladder-index. Then we can in time~$O(d^2 8^d \cdot n)$ decide whether $G$ contains
  a ladder of size at least~$\lfloor t/2 \rfloor$.
\end{theorem}
\begin{proof}
  We compute a degeneracy ordering $\mathbb G$ of $G$ and initialize the data
  structure~$R$ as per \Cref{lemma:R} in time $O(2^d n)$.

  Let $(A,B)$ induce a ladder of maximum size~$t$ in $G$, by \Cref{obs:ladder-degen-bound} we
  have that $t \leq 2d+1$. By \Cref{obs:ladder-lc}, there exists a vertex $u \in
  A \cup B$ such that $N^-(u)$ contains either $A' \defineq (a_i)_{i \geq t/2}$ or
  $B' \defineq (b_i)_{i \leq t/2}$. Wlog assume $A' \subseteq N^-(u)$ and let $k \defineq
  |A'|$. We guess $u$ in $O(n)$ time and $A' \subseteq N^-(u)$ in time $O
  (2^d)$. To verify that $A'$ can be completed into a ladder, we compute the
  data structure $\DataStruct_{A'}$ in time $O(k 2^k + dk)$ using \Cref{lemma:count}.

  Finally, we verify that there exists a sequence of subsets
  $A'_1 \subset A'_2 \subset \ldots \subset A'_k = A'$ where $\DataStruct_{A'}[A'_k] > 0$ for all $i \in [k]$; as each lookup in~$\DataStruct_{A'}$ has cost equal to the size of the
  query set this will take time proportional to
  $
      \sum_{i=0}^k i {k \choose i} = k 2^{k-1}
  $
  in the worst case (where we have to query all subsets of~$A'$ before finding the sequence). Since~$k \defineq \floor{t/2}$, the total running time of this algorithm is
  \[
    O(2^d n) + O\Big(2^d n \cdot (k 2^k + dk) \cdot k 2^{k-1}\Big)
    = O\Big(2^d k 2^k (k 2^k + dk) \cdot n \Big).
  \]
  We can simplify this expression further by using that~$k \leq d$ which leads us to the claimed running time of~$O(d^2 8^d \cdot n)$
\end{proof}

\section{Implementation and experiments}
\label{sec:experiments}

Based on the above theoretical ideas, we implemented algorithms\footnote
{Source code available under \url
{https://github.com/microgravitas/mantis-shrimp/}} to compute the
VC-dimension, find the largest biclique, co-matchings (within an additive
error of 1) and ladder (within a factor~$2$). The last three algorithms
all simply check the left-neighbourhood for the respective structure. Aside
from  optimisations of the involved data structures we will not describe
these algorithms in further detail.

We observe that for practical purposes the data structure $R$ can be computed
progressively: if we know that our algorithm currently only needs to compute
$\DataStruct_S$ from~$R$ (as per \Cref{lemma:count}) with $|S| = k$ ($k \leq d$), then it is enough to
only count sets of size $\leq k$ in~$R$. We can achieve this in time $O(
{d \choose k} n)$, which is far preferable to using $O(2^d n)$ time to
insert all left-neighbourhood subsets into $R$. If $k$ remains much smaller
than~$d$, this improves our running time and space consumption
substantially.

The second important optimisation regards subset dictionaries. While tries are
useful in our theoretical analysis, in practice we opted to use bitsets for
the data structures $\DataStruct_S$, as their universe $S$ can assumed to be small.
Bitsets also allow for a very concise and fast implementation of the
fast \Mobius inversion, which needs to happen very frequently inside the hot
loop of the search algorithms.

The algorithm to compute the VC-dimension includes a few simple optimisations
that vastly improved its performance. Note that if we are currently searching
for a shattered set of size~$k$, then a candidate vertex for a shattered set
of size~$k$ must have at least~${k-1 \choose i-1}$ neighbours of degree at
least~$i$, for~$1 \leq i \leq k-1$. Our algorithm recomputes the set of
remaining candidates each time it finds a larger shattered set. The
(progressive) computation of the data structure $R$ can then also be
restricted to only those left-neighbourhoods subsets which only contain
candidate vertices.

Accordingly, the algorithm performs well if it finds large shattered sets
fast. To that end, it first only looks at $k$-subsets of left-neigbhourhoods
of single vertices. Once that search is exhausted, it considers
left-neighbourhoods of pairs, then triplets, \etc up to $\lceil \log
d+1 \rceil$ vertices (as per \Cref{lemma:shattered-small-lc}).
As this search is very expensive once we need to consider the joint
left-neighbourhood of several vertices, the algorithm estimates the work
needed and compares it against simply brute-forcing all $k$-subsets of the
remaining candidates. Since the number of candidates shrinks quite quickly in
practice, the algorithm usually concludes with such a final exhaustive
search.

\subsection{Results}

\begin{figure}[thb]
  \centering
  \includegraphics[scale=.50]{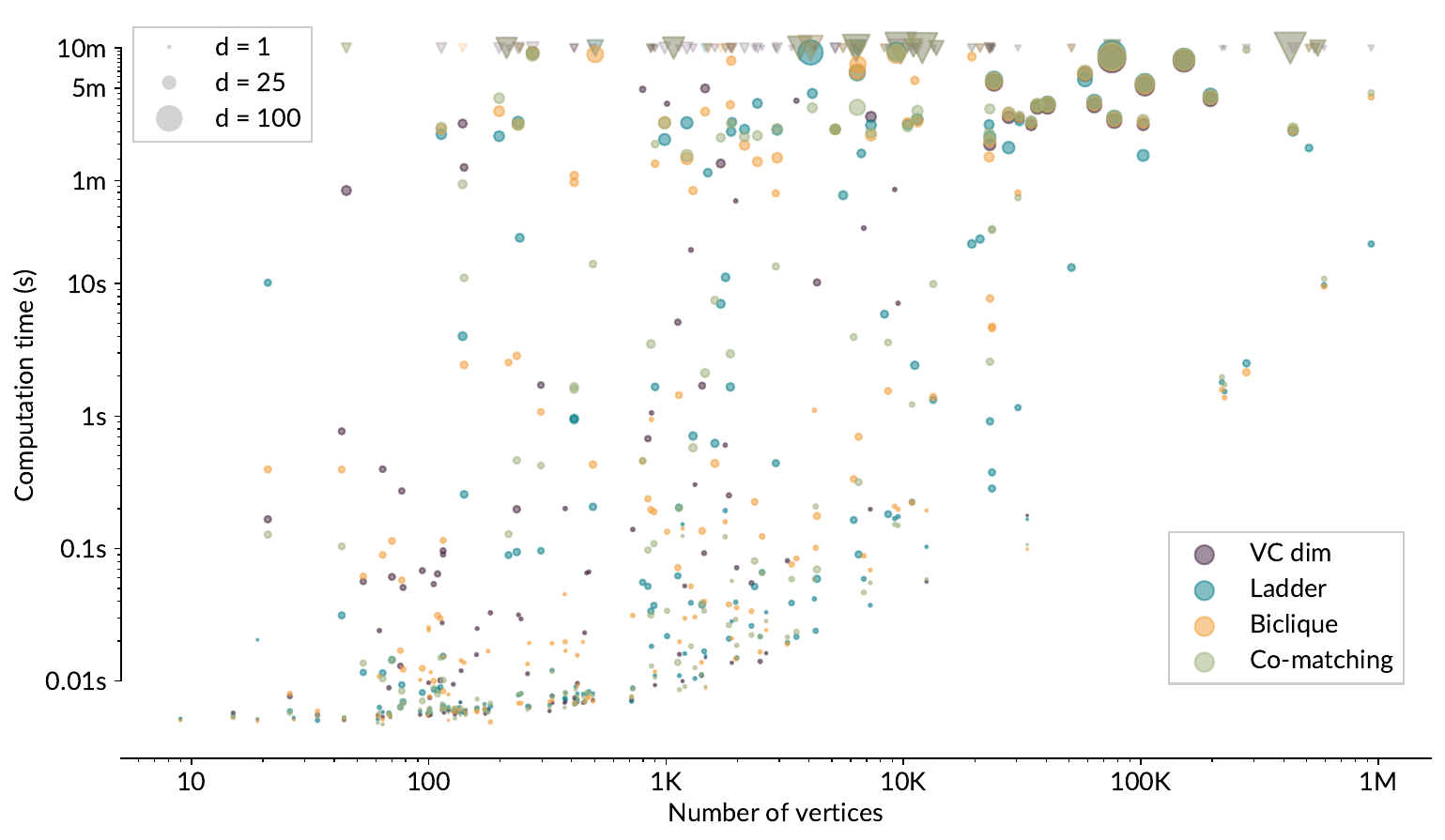}
  \caption{\label{fig:running-times}Running times of all four algorithms on a collection of~206 networks.
    The size of the circles indicates the degeneracy of the networks, triangles indicate that program timed out on the network after~10 minutes.
  }
\end{figure}

We implemented all four algorithms in Rust and tested them on a diverse
collection of~206 networks\footnote{\url
{https://github.com/microgravitas/network-corpus}}, using a PC with a AMD
Ryzen 3 2200G CPU and 24 GB RAM. The primary goal of our experiments was to
verify that the data structures and algorithms in this paper could be of
practical use, therefore we ran each algorithm only once per network\footnote
{The variance in running times was on the orders of seconds.} and timed out
after 10 minutes.

Of all the four measures, computing the VC-dimension is, unsuprisingly, the most computationally
challenging and the program timed out or ran out of memory for networks larger than a few ten-thousand nodes or of degeneracy higher than 24. The broad summary of the results looks as follows:
\begin{center}
\begin{tabular}{llll}
  Statistics & Completed & Max size ($n$) & Max degeneracy \\ \midrule
  VC-dimension & 126 & 33266 (\texttt{BioGrid-Chemicals}) & 24 (\texttt{wafa-eies}) \\
  Biclique & 176 & 935591 (\texttt{teams}) & 191 (\texttt{BioGrid-All}) \\
  Co-matching & 179 & 935591 (\texttt{teams}) & 255 (\texttt{dogster\_friendships}) \\
  Ladder index & 187 & 935591 (\texttt{teams}) & 191 (\texttt{BioGrid-All}) \\
\end{tabular}
\end{center}
\Cref{fig:running-times} visualizes these results in more detail.

\begin{figure}[htb]
  \centering
  \vspace*{-12pt}
  \includegraphics[scale=.55]{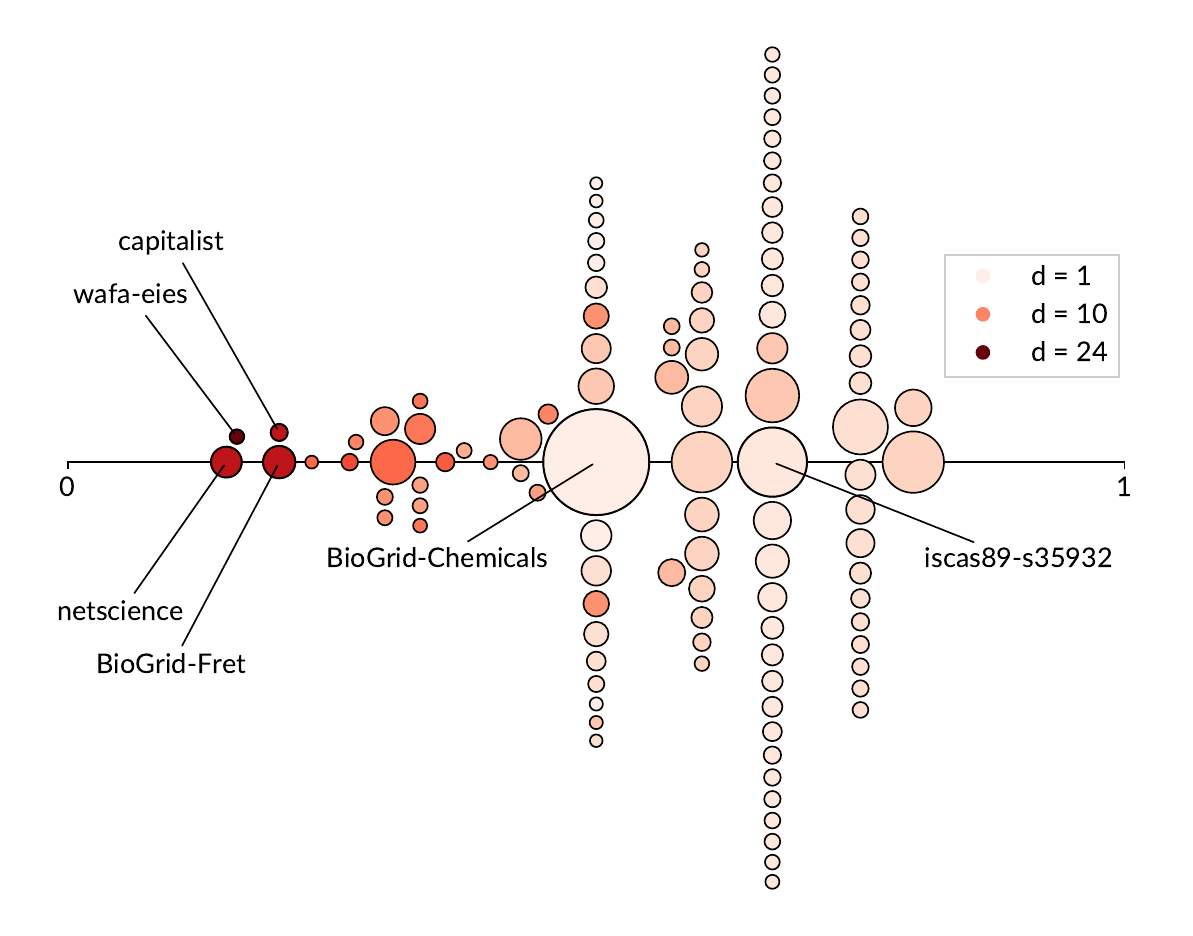}
  \vspace*{-16pt}
  \caption{\label{fig:vc-dim-degen}VC-dimension of networks normalized by their degeneracy${}+ 1$. Networks with large degeneracy tend towards the left, meaning that the VC-dimension does not increase proportionally to the degeneracy.
  }
\end{figure}

We are also interested in typical values of the VC-dimension of networks and
how it compares to the degeneracy. This topic deserves a deeper
investigation, but we can report some preliminary results here for those
networks where our program terminated before the timeout. In \Cref{fig:vc-dim-degen} we normalised the VC-dimension by the degeneracy-plus-one,
so values close to one indicate that the VC-dimension is on the order of the
degeneracy while values close to zero indicate that it is much
smaller than the degeneracy. We see a clear tendency that networks with
larger degeneracy tend towards zero, which we interpret as the VC-dimension ``growing slower'' than the degeneracy in typical networks.

\section{Conclusion}
\label{sec:conclusion}

On the theoretical side,
we outlined a general bipartite pattern-finding and -counting algorithm in
degenerate graphs. Its running time crucially depends on two complexity
measures of patterns, namely the left-covering number and the half-ordering
asymmetry. These general algorithms can be further improved for specific
patterns, which we exemplify for shattered set, ladder, co-matching and
biclique patterns.  Our results also include improved running times  when the input  graphs are of bounded degeneracy.

On the experimental side, we demonstrate that this style of algorithm is
feasible and practical for computation on real-world networks, which often
exhibit low degeneracy. The experiments also suggest that the VC-dimension of
networks tends to be a very small parameter, which makes it an interesting
target for the development of fast algorithms that exploit low VC-dimension.

 \appendix
\addtocontents{toc}{\protect\setcounter{tocdepth}{0}}

 \section{Complete results}

 \begin{landscape}
 \def\BioGrid{\textcolor{gray}{BG}} \def\HIV{\textcolor{gray}{HIV}} \def\HPV{\textcolor{gray}{HPV}}\def\EColi{\textcolor{gray}{E.-Coli}} \def\AC{\textcolor{gray}{AC}} \def\EAT{\textcolor{gray}{edinburgh\_assoc.\_thes.}} \def\ATC{\textcolor{gray}{ATC}} \def\DDiscoideum{\textcolor{gray}{D.-Discoideum}} \def\CAlbicans{\textcolor{gray}{C.-Albicans}} \def\CElegans{\textcolor{gray}{C.-Elegans}} \def\DMelanogaster{\textcolor{gray}{D.-Melanogaster}} \def\ENidulans{\textcolor{gray}{E.-Nidulans}} \def\Fly{\textcolor{gray}{bn-fly-d.\_medulla\_1}} \def\HSV{\textcolor{gray}{HSV}} \def\Dosage{\textcolor{gray}{D.}} 

\noindent
For space reasons, some of the network names are abbreviated below. Abbreviated names are marked in gray.

\begingroup\small
\begin{longtable}{@{}>{\small}llrrrrrllllllll@{}}
\toprule
 &          &          &          &     &     & \multicolumn{4}{c}{Running time} & \multicolumn{4}{c}{Statistics} \\
 & $n$ & $m$ & $\delta$ & $\bar d$ & deg & $\Delta$ & VC dim & biclique & crown  & ladder & VC dim & biclique & crown & ladder \\
Network &  &  &  &  &  &  &  &  &  &  &  &  &  &  \\
\midrule
\endfirsthead
\toprule
 &          &          &          &     &     & \multicolumn{4}{c}{Running time} & \multicolumn{4}{c}{Statistics} \\
 Network  & $n$ & $m$ & $\delta$ & $\bar d$ & deg & $\Delta$ & VC dim & biclique & crown  & ladder & VC dim & biclique & crown & ladder \\
\midrule
\endhead
\midrule
\multicolumn{15}{r}{Continued on next page} \\
\midrule
\endfoot
\bottomrule
\endlastfoot
AS-oregon-1 & 11174 & 23409 & 2389 & 4.2 & 17 & 2389 & 600.09 & 342.08 & 175.43 & 2.42 & [5,18] & 12 & [13,14] & [17,35] \\
AS-oregon-2 & 11461 & 32730 & 2432 & 5.7 & 31 & 2432 & 600.40 & 167.57 & 201.43 & 174.27 & [6,32] & [7,31] & [7,32] & [7,63] \\
\BioGrid-\AC-Luminescence & 1840 & 2312 & 376 & 2.5 & 6 & 376 & 0.25 & 0.04 & 0.03 & 0.02 & 4 & 5 & 7 & [6,13] \\
\BioGrid-\AC-Ms & 40495 & 321887 & 2217 & 15.9 & 58 & 2217 & 217.75 & 224.87 & 227.38 & 227.77 & [4,59] & [4,58] & [4,59] & [4,117] \\
\BioGrid-\AC-Rna & 13765 & 42815 & 3572 & 6.2 & 54 & 3572 & 600.21 & 601.05 & 601.38 & 601.54 & [4,55] & [5,54] & [5,55] & [5,109] \\
\BioGrid-\AC-Western & 21028 & 64046 & 535 & 6.1 & 17 & 535 & 600.08 & 600.30 & 600.29 & 21.70 & [5,18] & [7,17] & [9,18] & [17,35] \\
\BioGrid-All & 75550 & 1316843 & 3620 & 34.9 & 191 & 3620 & 496.99 & 514.30 & 519.60 & 540.54 & [2,192] & [2,191] & [2,192] & [2,383] \\
\BioGrid-\ATC & 10417 & 47916 & 1341 & 9.2 & 26 & 1341 & 600.25 & 163.89 & 152.58 & 159.18 & [6,27] & [8,26] & [8,27] & [8,53] \\
\BioGrid-Biochemical-Activity & 8620 & 17746 & 427 & 4.1 & 11 & 427 & 600.11 & 1.55 & 3.59 & 0.18 & [5,12] & 8 & [7,8] & [11,23] \\
\BioGrid-Bos-Taurus & 454 & 424 & 27 & 1.9 & 3 & 27 & 0.02 & 0.01 & 0.01 & 0.01 & 2 & 3 & [3,4] & [3,7] \\
\BioGrid-\CElegans & 6394 & 23646 & 522 & 7.4 & 64 & 522 & 600.21 & 399.05 & 214.18 & 390.57 & [4,65] & [6,64] & [5,65] & [6,129] \\
\BioGrid-\CAlbicans-Sc5314 & 1121 & 1609 & 427 & 2.9 & 9 & 427 & 5.11 & 0.07 & 0.01 & 0.06 & 3 & 6 & 10 & [9,19] \\
\BioGrid-Canis-Familiaris & 143 & 125 & 90 & 1.7 & 2 & 90 & 0.01 & 0.01 & 0.01 & 0.01 & 2 & 2 & 3 & [2,5] \\
\BioGrid-Chemicals & 33266 & 28093 & 413 & 1.7 & 1 & 413 & 0.18 & 0.10 & 0.11 & 0.17 & 1 & [0,1] & 2 & [0,3] \\
\BioGrid-Co-Crystal-Structure & 2291 & 2021 & 92 & 1.8 & 5 & 92 & 0.05 & 0.03 & 0.02 & 0.03 & 3 & 3 & 6 & [3,7] \\
\BioGrid-Co-Fractionation & 11017 & 56354 & 187 & 10.2 & 83 & 187 & 600.67 & 602.19 & 601.34 & 601.14 & [4,84] & [5,83] & [5,84] & [5,167] \\
\BioGrid-Co-Localization & 3543 & 4452 & 63 & 2.5 & 6 & 63 & 240.56 & 0.08 & 0.02 & 0.02 & 3 & 5 & 7 & [6,13] \\
\BioGrid-Co-Purification & 4326 & 5970 & 1972 & 2.8 & 12 & 1972 & 10.21 & 0.18 & 0.07 & 0.06 & 4 & 6 & 13 & [12,25] \\
\BioGrid-Cricetulus-Griseus & 69 & 57 & 30 & 1.7 & 1 & 30 & 0.01 & 0.01 & 0.01 & 0.01 & 1 & [0,1] & 2 & [0,3] \\
\BioGrid-Danio-Rerio & 261 & 266 & 61 & 2.0 & 3 & 61 & 0.01 & 0.01 & 0.01 & 0.01 & 2 & 2 & 4 & [3,7] \\
\BioGrid-\DDiscoideum-Ax4 & 27 & 20 & 4 & 1.5 & 1 & 4 & 0.01 & 0.01 & 0.01 & 0.01 & 1 & [0,1] & 2 & [0,3] \\
\BioGrid-\Dosage-Growth-Defect & 1447 & 2193 & 213 & 3.0 & 5 & 213 & 0.09 & 0.04 & 0.03 & 0.02 & 4 & 4 & 6 & [5,11] \\
\BioGrid-\Dosage-Lethality & 1776 & 2289 & 392 & 2.6 & 4 & 392 & 0.60 & 0.16 & 0.12 & 0.19 & 3 & 4 & [4,5] & [4,9] \\
\BioGrid-\Dosage-Rescue & 3380 & 6444 & 75 & 3.8 & 7 & 75 & 600.01 & 0.08 & 0.06 & 0.04 & [4,8] & 6 & [7,8] & [7,15] \\
\BioGrid-\DMelanogaster & 9330 & 60556 & 303 & 13.0 & 83 & 303 & 600.32 & 538.71 & 553.50 & 566.57 & [4,84] & [5,83] & [5,84] & [5,167] \\
\BioGrid-\ENidulans-Fgsc-A4 & 64 & 62 & 44 & 1.9 & 2 & 44 & 0.01 & 0.01 & 0.00 & 0.01 & 2 & 2 & 3 & [2,5] \\
\BioGrid-\EColi-K12-Mg1655 & 1273 & 1889 & 58 & 3.0 & 5 & 58 & 17.95 & 0.05 & 0.02 & 0.04 & 3 & 4 & 6 & [5,11] \\
\BioGrid-\EColi-K12-W3110 & 4063 & 181620 & 1187 & 89.4 & 159 & 1187 & 600.95 & 600.94 & 600.90 & 556.39 & [3,160] & [3,159] & [3,160] & [3,319] \\
\BioGrid-Far-Western & 1199 & 1089 & 60 & 1.8 & 3 & 60 & 0.05 & 0.03 & 0.01 & 0.02 & 3 & 3 & 4 & [3,7] \\
\BioGrid-Fret & 1700 & 2395 & 51 & 2.8 & 19 & 51 & 80.71 & 600.09 & 126.45 & 7.04 & 4 & [11,19] & [17,18] & [19,39] \\
\BioGrid-Gallus-Gallus & 413 & 436 & 110 & 2.1 & 4 & 110 & 0.01 & 0.01 & 0.01 & 0.01 & 3 & 3 & [4,5] & [4,9] \\
\BioGrid-Glycine-Max & 44 & 39 & 13 & 1.8 & 2 & 13 & 0.01 & 0.01 & 0.01 & 0.01 & 2 & 2 & 3 & [2,5] \\
\BioGrid-Hepatitus-C-Virus & 136 & 134 & 133 & 2.0 & 1 & 133 & 0.01 & 0.01 & 0.01 & 0.01 & 1 & [0,1] & 2 & [0,3] \\
\BioGrid-Homo-Sapiens & 24093 & 369767 & 2882 & 30.7 & 71 & 2882 & 329.05 & 338.13 & 343.83 & 346.77 & [3,72] & [3,71] & [3,72] & [3,143] \\
\BioGrid-\HSV-1 & 178 & 208 & 40 & 2.3 & 3 & 40 & 0.01 & 0.02 & 0.01 & 0.01 & 3 & 3 & 4 & [3,7] \\
\BioGrid-\HSV-4 & 323 & 326 & 154 & 2.0 & 2 & 154 & 0.01 & 0.01 & 0.01 & 0.01 & 2 & 2 & [2,3] & [2,5] \\
\BioGrid-\HSV-5 & 121 & 107 & 27 & 1.8 & 1 & 27 & 0.01 & 0.01 & 0.01 & 0.01 & 1 & [0,1] & 2 & [0,3] \\
\BioGrid-\HSV-8 & 716 & 691 & 119 & 1.9 & 3 & 119 & 0.01 & 0.01 & 0.01 & 0.01 & 2 & 3 & [3,4] & [3,7] \\
\BioGrid-\HIV-1 & 1138 & 1319 & 324 & 2.3 & 3 & 324 & 0.02 & 0.02 & 0.01 & 0.01 & 3 & 3 & 4 & [3,7] \\
\BioGrid-\HIV-2 & 19 & 15 & 6 & 1.6 & 1 & 6 & 0.01 & 0.00 & 0.01 & 0.02 & 1 & [0,1] & 2 & [0,3] \\
\BioGrid-\HPV-16 & 173 & 186 & 93 & 2.2 & 2 & 93 & 0.01 & 0.01 & 0.01 & 0.01 & 2 & 2 & [2,3] & [2,5] \\
Cannes2013 & 438089 & 835892 & 15169 & 3.8 & 27 & 15169 & 600.79 & 144.70 & 149.74 & 142.22 & [5,28] & [6,27] & [6,28] & [6,55] \\
CoW-interstate & 182 & 319 & 25 & 3.5 & 4 & 25 & 0.03 & 0.00 & 0.01 & 0.01 & 3 & 4 & [3,4] & [4,9] \\
DNC-emails & 1866 & 4384 & 402 & 4.7 & 17 & 402 & 600.11 & 224.12 & 2.95 & 1.66 & [5,18] & 10 & [16,17] & [17,35] \\
EU-email-core & 986 & 16064 & 345 & 32.6 & 34 & 345 & 600.86 & 164.82 & 163.41 & 122.48 & [5,35] & [6,34] & [6,35] & [6,69] \\
JDK\_dependency & 6434 & 53658 & 5923 & 16.7 & 65 & 5923 & 600.41 & 453.02 & 600.99 & 601.73 & [4,66] & [5,65] & [5,66] & [5,131] \\
JUNG-javax & 6120 & 50290 & 5655 & 16.4 & 65 & 5655 & 600.35 & 600.93 & 600.95 & 602.14 & [4,66] & [5,65] & [5,66] & [5,131] \\
NYClimateMarch2014 & 102378 & 327080 & 14687 & 6.4 & 34 & 14687 & 158.65 & 165.32 & 170.79 & 92.84 & [5,35] & [5,34] & [5,35] & [4,69] \\
NZ\_legal & 2141 & 15739 & 429 & 14.7 & 25 & 429 & 600.36 & 110.96 & 128.55 & 146.23 & [6,26] & [7,25] & [7,26] & [7,51] \\
Noordin-terror-loc & 127 & 190 & 18 & 3.0 & 3 & 18 & 0.01 & 0.01 & 0.01 & 0.01 & 3 & 3 & [3,4] & [3,7] \\
Noordin-terror-orgas & 129 & 181 & 21 & 2.8 & 3 & 21 & 0.01 & 0.01 & 0.01 & 0.01 & 3 & 3 & [3,4] & [3,7] \\
Noordin-terror-relation & 70 & 251 & 28 & 7.2 & 11 & 28 & 0.06 & 0.11 & 0.01 & 0.01 & 4 & 6 & 12 & [11,23] \\
ODLIS & 2900 & 16377 & 592 & 11.3 & 12 & 592 & 600.03 & 48.04 & 13.49 & 0.44 & [5,13] & 6 & [9,10] & [12,25] \\
Opsahl-forum & 899 & 7036 & 128 & 15.7 & 14 & 128 & 600.02 & 80.31 & 113.04 & 1.66 & [5,15] & 5 & [6,7] & [14,29] \\
Opsahl-socnet & 1899 & 13838 & 255 & 14.6 & 20 & 255 & 600.12 & 600.33 & 600.29 & 166.27 & [6,21] & [7,20] & [8,21] & [20,41] \\
StackOverflow-tags & 115 & 245 & 16 & 4.3 & 6 & 16 & 0.09 & 0.01 & 0.01 & 0.01 & 3 & 4 & 7 & [6,13] \\
Y2H\_union & 1966 & 2705 & 89 & 2.8 & 4 & 89 & 42.11 & 0.01 & 0.03 & 0.04 & 3 & 4 & 5 & [4,9] \\
Yeast & 2361 & 7182 & 66 & 6.1 & 10 & 66 & 600.04 & 0.23 & 0.08 & 0.05 & [4,11] & 7 & [9,10] & [10,21] \\
actor\_movies & 511463 & 1470404 & 646 & 5.7 & 14 & 646 & 600.36 & 600.51 & 600.46 & 105.76 & [5,15] & [7,14] & [6,15] & [14,29] \\
advogato & 5155 & 39285 & 803 & 15.2 & 25 & 803 & 600.41 & 145.81 & 146.46 & 146.33 & [5,26] & [6,25] & [6,26] & [6,51] \\
airlines & 235 & 1297 & 130 & 11.0 & 13 & 130 & 0.20 & 2.85 & 0.46 & 0.09 & 5 & 8 & [11,12] & [13,27] \\
american\_revolution & 141 & 160 & 59 & 2.3 & 3 & 59 & 0.01 & 0.01 & 0.01 & 0.01 & 3 & 3 & [2,3] & [3,7] \\
as-22july06 & 22963 & 48436 & 2390 & 4.2 & 25 & 2390 & 600.20 & 90.18 & 137.08 & 158.60 & [6,26] & [8,25] & [10,26] & [10,51] \\
as20000102 & 6474 & 12572 & 1458 & 3.9 & 12 & 1458 & 600.07 & 0.70 & 0.32 & 0.09 & [5,13] & 9 & [10,11] & [12,25] \\
autobahn & 374 & 478 & 5 & 2.6 & 2 & 5 & 0.01 & 0.05 & 0.01 & 0.01 & 2 & 2 & 3 & [2,5] \\
bahamas & 219856 & 246291 & 14902 & 2.2 & 6 & 14902 & 600.12 & 1.59 & 1.97 & 1.80 & [3,7] & 6 & [3,4] & [6,13] \\
bergen & 53 & 272 & 32 & 10.3 & 9 & 32 & 0.06 & 0.06 & 0.01 & 0.01 & 4 & 6 & [8,9] & [9,19] \\
bitcoin-otc-negative & 1606 & 3259 & 227 & 4.1 & 16 & 227 & 600.05 & 0.44 & 7.48 & 0.62 & [4,17] & 16 & [5,6] & [16,33] \\
bitcoin-otc-positive & 5573 & 18591 & 788 & 6.7 & 20 & 788 & 600.14 & 600.45 & 600.24 & 46.50 & [6,21] & [8,20] & [10,21] & [20,41] \\
\Fly & 1781 & 8911 & 927 & 10.0 & 18 & 927 & 600.04 & 600.17 & 600.15 & 11.17 & [5,19] & [8,18] & [8,19] & [18,37] \\
bn-mouse\_retina\_1 & 1076 & 90811 & 744 & 168.8 & 121 & 744 & 600.37 & 600.48 & 600.42 & 600.68 & [3,122] & [3,121] & [3,122] & [3,243] \\
boards\_gender\_1m & 4134 & 19993 & 88 & 9.7 & 25 & 88 & 600.05 & 600.30 & 212.44 & 272.75 & [4,26] & [13,25] & 26 & [25,51] \\
boards\_gender\_2m & 4220 & 5598 & 45 & 2.7 & 4 & 45 & 600.10 & 1.10 & 0.06 & 0.04 & [3,5] & 4 & [3,4] & [4,9] \\
ca-CondMat & 23133 & 93439 & 279 & 8.1 & 25 & 279 & 600.05 & 600.65 & 208.73 & 600.65 & [4,26] & [14,25] & 26 & [18,51] \\
ca-HepPh & 12006 & 118489 & 491 & 19.7 & 238 & 491 & 600.21 & 600.33 & 600.13 & 600.17 & [3,239] & [3,238] & [3,239] & [3,477] \\
capitalist & 139 & 1071 & 91 & 15.4 & 19 & 91 & 161.57 & 600.32 & 56.29 & 4.01 & 4 & [12,19] & [17,18] & [19,39] \\
celegans & 297 & 2148 & 134 & 14.5 & 10 & 134 & 1.72 & 1.08 & 0.42 & 0.10 & 5 & 7 & [8,9] & [10,21] \\
chess & 7301 & 55899 & 181 & 15.3 & 29 & 181 & 182.00 & 130.97 & 138.03 & 157.14 & [6,30] & [6,29] & [6,30] & [6,59] \\
chicago & 1467 & 1298 & 12 & 1.8 & 1 & 12 & 0.02 & 0.01 & 0.01 & 0.02 & 1 & [0,1] & 2 & [0,3] \\
cit-HepPh & 34546 & 420877 & 846 & 24.4 & 30 & 846 & 156.96 & 163.15 & 169.08 & 164.14 & [4,31] & [4,30] & [4,31] & [4,61] \\
cit-HepTh & 27769 & 352285 & 2468 & 25.4 & 37 & 2468 & 180.07 & 189.90 & 194.29 & 106.40 & [4,38] & [4,37] & [4,38] & [3,75] \\
codeminer & 724 & 1015 & 55 & 2.8 & 4 & 55 & 0.14 & 0.03 & 0.01 & 0.01 & 3 & 4 & [4,5] & [4,9] \\
columbia-mobility & 863 & 4147 & 228 & 9.6 & 9 & 228 & 600.11 & 0.20 & 0.03 & 0.03 & [4,10] & 6 & 10 & [9,19] \\
columbia-social & 863 & 7724 & 545 & 17.9 & 18 & 545 & 600.11 & 600.29 & 3.51 & 600.13 & [4,19] & [12,18] & [18,19] & [14,37] \\
cora\_citation & 23166 & 89157 & 377 & 7.7 & 13 & 377 & 600.11 & 7.72 & 2.57 & 0.91 & [5,14] & 9 & [10,11] & [13,27] \\
countries & 592414 & 624402 & 110602 & 2.1 & 6 & 110602 & 600.14 & 9.49 & 10.88 & 9.75 & [4,7] & 5 & [4,5] & [6,13] \\
cpan-authors & 839 & 2112 & 327 & 5.0 & 9 & 327 & 0.68 & 0.24 & 0.10 & 0.05 & 5 & 6 & [8,9] & [9,19] \\
digg & 30398 & 86312 & 285 & 5.7 & 9 & 285 & 600.07 & 48.32 & 44.61 & 1.16 & [4,10] & 4 & [5,6] & [9,19] \\
diseasome & 1419 & 2738 & 84 & 3.9 & 11 & 84 & 1.70 & 0.14 & 0.04 & 0.04 & 4 & 6 & 12 & [11,23] \\
dogster\_friendships & 426820 & 8546581 & 46505 & 40.0 & 255 & 46505 & 600.38 & 600.57 & 600.00 & 600.44 & [1,256] & [0,255] & [1,256] & [0,511] \\
dolphins & 62 & 159 & 12 & 5.1 & 4 & 12 & 0.02 & 0.01 & 0.01 & 0.01 & 3 & 3 & 5 & [4,9] \\
ecoli-transcript & 423 & 578 & 74 & 2.7 & 3 & 74 & 0.01 & 0.02 & 0.01 & 0.01 & 3 & 3 & [3,4] & [3,7] \\
\EAT & 23132 & 297094 & 1062 & 25.7 & 34 & 1062 & 112.09 & 119.67 & 123.49 & 128.27 & [3,35] & [3,34] & [3,35] & [3,69] \\
email-Enron & 36692 & 183831 & 1383 & 10.0 & 43 & 1383 & 213.38 & 219.83 & 223.53 & 219.47 & [4,44] & [4,43] & [4,44] & [4,87] \\
euroroad & 1174 & 1417 & 10 & 2.4 & 2 & 10 & 0.01 & 0.14 & 0.12 & 0.15 & 2 & 2 & 3 & [2,5] \\
eva-corporate & 7253 & 6711 & 552 & 1.9 & 3 & 552 & 0.20 & 0.07 & 0.06 & 0.04 & 3 & 3 & 4 & [3,7] \\
exnet-water & 1893 & 2416 & 10 & 2.6 & 2 & 10 & 0.01 & 0.02 & 0.06 & 0.03 & 2 & 2 & 3 & [2,5] \\
facebook-links & 63731 & 817090 & 1098 & 25.6 & 52 & 1098 & 221.29 & 229.98 & 233.75 & 237.25 & [3,53] & [3,52] & [3,53] & [3,105] \\
foldoc & 13356 & 91471 & 728 & 13.7 & 12 & 728 & 600.10 & 1.39 & 9.93 & 1.33 & [5,13] & 12 & [9,10] & [12,25] \\
foodweb-caribbean & 492 & 3313 & 196 & 13.5 & 13 & 196 & 600.10 & 0.43 & 14.08 & 0.21 & [4,14] & 12 & [7,8] & [13,27] \\
foodweb-otago & 141 & 832 & 45 & 11.8 & 14 & 45 & 75.41 & 2.43 & 11.06 & 0.26 & 4 & 12 & [7,8] & [14,29] \\
football & 115 & 613 & 12 & 10.7 & 8 & 12 & 0.10 & 0.12 & 0.01 & 0.02 & 4 & 4 & 9 & [8,17] \\
google+ & 23628 & 39194 & 2761 & 3.3 & 12 & 2761 & 600.08 & 4.60 & 25.37 & 0.38 & [6,13] & 9 & [7,8] & [12,25] \\
gowalla & 196591 & 950327 & 14730 & 9.7 & 51 & 14730 & 246.26 & 254.67 & 256.89 & 263.96 & [3,52] & [3,51] & [3,52] & [3,103] \\
haggle & 274 & 2124 & 101 & 15.5 & 39 & 101 & 600.11 & 551.34 & 533.96 & 550.92 & [4,40] & [8,39] & [8,40] & [8,79] \\
hex & 331 & 930 & 6 & 5.6 & 3 & 6 & 0.01 & 0.02 & 0.01 & 0.01 & 3 & 2 & [3,4] & [3,7] \\
hypertext\_2009 & 113 & 2196 & 98 & 38.9 & 28 & 98 & 600.15 & 146.97 & 150.88 & 134.22 & [5,29] & [9,28] & [9,29] & [9,57] \\
ia-email-univ & 1133 & 5451 & 71 & 9.6 & 11 & 71 & 600.06 & 1.44 & 0.20 & 0.20 & [4,12] & 6 & 12 & [11,23] \\
ia-infect-dublin & 410 & 2765 & 50 & 13.5 & 17 & 50 & 600.10 & 65.59 & 1.60 & 0.94 & [4,18] & 9 & [16,17] & [17,35] \\
ia-reality & 6809 & 7680 & 261 & 2.3 & 5 & 261 & 26.27 & 0.09 & 0.05 & 0.06 & 4 & 4 & [5,6] & [5,11] \\
infectious & 410 & 2765 & 50 & 13.5 & 17 & 50 & 600.10 & 58.09 & 1.67 & 0.96 & [4,18] & 9 & [16,17] & [17,35] \\
iscas89-s1196 & 377 & 537 & 16 & 2.8 & 2 & 16 & 0.01 & 0.02 & 0.01 & 0.01 & 2 & 2 & [2,3] & [2,5] \\
iscas89-s1238 & 416 & 625 & 18 & 3.0 & 2 & 18 & 0.01 & 0.01 & 0.01 & 0.01 & 2 & 2 & [2,3] & [2,5] \\
iscas89-s13207 & 2492 & 3406 & 37 & 2.7 & 4 & 37 & 0.01 & 0.02 & 0.02 & 0.02 & 4 & 4 & [4,5] & [4,9] \\
iscas89-s1423 & 423 & 554 & 17 & 2.6 & 2 & 17 & 0.01 & 0.01 & 0.01 & 0.01 & 2 & 2 & [2,3] & [2,5] \\
iscas89-s1488 & 463 & 779 & 53 & 3.4 & 3 & 53 & 0.07 & 0.01 & 0.01 & 0.01 & 3 & 3 & [3,4] & [3,7] \\
iscas89-s1494 & 473 & 796 & 56 & 3.4 & 3 & 56 & 0.07 & 0.01 & 0.01 & 0.01 & 3 & 3 & [3,4] & [3,7] \\
iscas89-s15850 & 3247 & 4004 & 25 & 2.5 & 4 & 25 & 0.08 & 0.02 & 0.02 & 0.02 & 3 & 4 & [3,4] & [4,9] \\
iscas89-s208 & 61 & 67 & 8 & 2.2 & 2 & 8 & 0.01 & 0.01 & 0.00 & 0.01 & 2 & 2 & [2,3] & [2,5] \\
iscas89-s27 & 9 & 8 & 3 & 1.8 & 1 & 3 & 0.01 & 0.01 & 0.01 & 0.01 & 1 & [0,1] & 2 & [0,3] \\
iscas89-s298 & 92 & 131 & 11 & 2.8 & 2 & 11 & 0.01 & 0.01 & 0.01 & 0.01 & 2 & 2 & [2,3] & [2,5] \\
iscas89-s344 & 100 & 122 & 9 & 2.4 & 2 & 9 & 0.01 & 0.02 & 0.01 & 0.01 & 2 & 2 & [2,3] & [2,5] \\
iscas89-s349 & 102 & 127 & 9 & 2.5 & 2 & 9 & 0.01 & 0.01 & 0.01 & 0.01 & 2 & 2 & [2,3] & [2,5] \\
iscas89-s35932 & 12515 & 15961 & 1440 & 2.6 & 2 & 1440 & 0.06 & 0.19 & 0.06 & 0.10 & 2 & [0,1] & 3 & [2,5] \\
iscas89-s382 & 116 & 168 & 18 & 2.9 & 2 & 18 & 0.01 & 0.02 & 0.01 & 0.01 & 2 & 2 & [2,3] & [2,5] \\
iscas89-s38417 & 9500 & 10635 & 39 & 2.2 & 4 & 39 & 7.11 & 0.20 & 0.15 & 0.17 & 4 & 2 & [4,5] & [4,9] \\
iscas89-s38584 & 9193 & 12573 & 54 & 2.7 & 4 & 54 & 51.37 & 0.21 & 0.15 & 0.17 & 3 & 4 & [3,4] & [4,9] \\
iscas89-s386 & 114 & 200 & 23 & 3.5 & 3 & 23 & 0.03 & 0.01 & 0.01 & 0.01 & 2 & 3 & [2,3] & [3,7] \\
iscas89-s400 & 121 & 182 & 19 & 3.0 & 2 & 19 & 0.01 & 0.01 & 0.01 & 0.01 & 2 & 2 & [2,3] & [2,5] \\
iscas89-s420 & 129 & 145 & 9 & 2.2 & 2 & 9 & 0.01 & 0.01 & 0.01 & 0.01 & 2 & 2 & [2,3] & [2,5] \\
iscas89-s444 & 134 & 206 & 19 & 3.1 & 2 & 19 & 0.01 & 0.01 & 0.01 & 0.01 & 2 & 2 & 3 & [2,5] \\
iscas89-s510 & 172 & 251 & 12 & 2.9 & 2 & 12 & 0.01 & 0.01 & 0.01 & 0.01 & 2 & 2 & [2,3] & [2,5] \\
iscas89-s526 & 160 & 270 & 12 & 3.4 & 3 & 12 & 0.02 & 0.01 & 0.01 & 0.01 & 3 & 3 & [2,3] & [3,7] \\
iscas89-s526n & 159 & 268 & 12 & 3.4 & 3 & 12 & 0.02 & 0.01 & 0.01 & 0.01 & 3 & 3 & [3,4] & [3,7] \\
iscas89-s5378 & 1411 & 1639 & 10 & 2.3 & 3 & 10 & 0.01 & 0.01 & 0.01 & 0.01 & 3 & 2 & [3,4] & [3,7] \\
iscas89-s641 & 100 & 144 & 12 & 2.9 & 3 & 12 & 0.01 & 0.03 & 0.01 & 0.01 & 3 & 2 & [2,3] & [3,7] \\
iscas89-s713 & 137 & 180 & 12 & 2.6 & 3 & 12 & 0.01 & 0.01 & 0.01 & 0.01 & 3 & 2 & [2,3] & [3,7] \\
iscas89-s820 & 239 & 480 & 48 & 4.0 & 3 & 48 & 0.03 & 0.01 & 0.01 & 0.01 & 3 & 3 & [3,4] & [3,7] \\
iscas89-s832 & 245 & 498 & 49 & 4.1 & 3 & 49 & 0.03 & 0.01 & 0.01 & 0.01 & 3 & 3 & [3,4] & [3,7] \\
iscas89-s838 & 265 & 301 & 12 & 2.3 & 2 & 12 & 0.01 & 0.02 & 0.01 & 0.01 & 2 & 2 & [2,3] & [2,5] \\
iscas89-s9234 & 1985 & 2370 & 18 & 2.4 & 4 & 18 & 0.07 & 0.04 & 0.02 & 0.01 & 3 & 4 & [3,4] & [4,9] \\
iscas89-s953 & 332 & 454 & 12 & 2.7 & 2 & 12 & 0.01 & 0.01 & 0.01 & 0.01 & 2 & 2 & [2,3] & [2,5] \\
jazz & 198 & 2742 & 100 & 27.7 & 29 & 100 & 600.04 & 200.16 & 250.91 & 129.73 & [5,30] & [12,29] & [13,30] & [11,59] \\
karate & 34 & 78 & 17 & 4.6 & 4 & 17 & 0.01 & 0.01 & 0.01 & 0.01 & 3 & 3 & 5 & [4,9] \\
lederberg & 8324 & 41532 & 1103 & 10.0 & 15 & 1103 & 600.11 & 600.11 & 600.08 & 5.88 & [5,16] & [7,15] & [9,16] & [15,31] \\
lesmiserables & 77 & 254 & 36 & 6.6 & 9 & 36 & 0.27 & 0.06 & 0.01 & 0.01 & 3 & 6 & 10 & [9,19] \\
link-pedigree & 898 & 1125 & 14 & 2.5 & 2 & 14 & 0.01 & 0.01 & 0.01 & 0.01 & 2 & 2 & [2,3] & [2,5] \\
linux & 30834 & 213217 & 9338 & 13.8 & 23 & 9338 & 178.38 & 179.08 & 184.87 & 169.29 & [7,24] & [7,23] & [7,24] & [7,47] \\
livemocha & 104103 & 2193083 & 2980 & 42.1 & 92 & 2980 & 308.64 & 318.47 & 325.14 & 326.65 & [2,93] & [2,92] & [2,93] & [2,185] \\
loc-brightkite\_edges & 58228 & 214078 & 1134 & 7.4 & 52 & 1134 & 381.48 & 385.52 & 393.61 & 346.98 & [5,53] & [5,52] & [5,53] & [5,105] \\
location & 225486 & 293697 & 12189 & 2.6 & 5 & 12189 & 600.30 & 1.38 & 1.73 & 1.53 & [4,6] & 5 & [4,5] & [5,11] \\
marvel & 19428 & 96662 & 1625 & 10.0 & 18 & 1625 & 600.14 & 516.95 & 600.19 & 19.96 & [6,19] & 15 & [8,19] & [18,37] \\
mg\_casino & 109 & 326 & 94 & 6.0 & 9 & 94 & 0.06 & 0.03 & 0.01 & 0.01 & 3 & 5 & 10 & [9,19] \\
mg\_forrestgump & 94 & 271 & 89 & 5.8 & 8 & 89 & 0.07 & 0.01 & 0.01 & 0.01 & 3 & 4 & 9 & [8,17] \\
mg\_godfatherII & 78 & 219 & 34 & 5.6 & 8 & 34 & 0.05 & 0.01 & 0.01 & 0.01 & 3 & 5 & 9 & [8,17] \\
mg\_watchmen & 76 & 201 & 33 & 5.3 & 7 & 33 & 0.01 & 0.02 & 0.01 & 0.01 & 3 & 4 & 8 & [7,15] \\
minnesota & 2642 & 3303 & 5 & 2.5 & 2 & 5 & 0.02 & 0.02 & 0.03 & 0.03 & 2 & 2 & 3 & [2,5] \\
moreno\_health & 2539 & 10455 & 27 & 8.2 & 7 & 27 & 600.10 & 0.12 & 0.07 & 0.07 & [4,8] & 5 & [7,8] & [7,15] \\
mousebrain & 213 & 16089 & 205 & 151.1 & 111 & 205 & 600.28 & 600.47 & 600.38 & 600.43 & [3,112] & [3,111] & [3,112] & [3,223] \\
movielens\_1m & 9746 & 1000209 & 3428 & 205.3 & 255 & 3428 & 600.76 & 600.79 & 600.75 & 600.81 & [2,256] & [2,255] & [2,256] & [2,511] \\
movies & 101 & 192 & 19 & 3.8 & 3 & 19 & 0.01 & 0.01 & 0.01 & 0.01 & 3 & 2 & [3,4] & [3,7] \\
muenchen-bahn & 447 & 578 & 13 & 2.6 & 2 & 13 & 0.01 & 0.02 & 0.01 & 0.01 & 2 & [0,1] & [2,3] & [2,5] \\
munin & 1324 & 1397 & 66 & 2.1 & 3 & 66 & 0.30 & 0.03 & 0.01 & 0.01 & 2 & 3 & [2,3] & [3,7] \\
netscience & 1461 & 2742 & 34 & 3.8 & 19 & 34 & 297.79 & 198.83 & 2.11 & 600.05 & 3 & 10 & 20 & [12,39] \\
offshore & 278877 & 505965 & 37336 & 3.6 & 13 & 37336 & 600.11 & 2.14 & 583.20 & 2.50 & [5,14] & 13 & [9,10] & [13,27] \\
openflights & 2939 & 15677 & 242 & 10.7 & 28 & 242 & 600.11 & 89.14 & 147.54 & 144.73 & [5,29] & [7,28] & [7,29] & [7,57] \\
p2p-Gnutella04 & 10876 & 39994 & 103 & 7.4 & 7 & 103 & 600.11 & 0.22 & 1.22 & 0.22 & [4,8] & 7 & [5,6] & [7,15] \\
panama & 556686 & 702437 & 7015 & 2.5 & 62 & 7015 & 600.49 & 601.18 & 601.09 & 601.14 & [4,63] & [6,62] & [6,63] & [6,125] \\
paradise & 542102 & 794545 & 35359 & 2.9 & 23 & 35359 & 600.21 & 601.52 & 600.31 & 601.57 & [5,24] & [22,23] & [6,24] & [23,47] \\
photoviz\_dynamic & 376 & 610 & 29 & 3.2 & 4 & 29 & 0.20 & 0.02 & 0.01 & 0.01 & 3 & 3 & [3,4] & [4,9] \\
pigs & 492 & 592 & 39 & 2.4 & 2 & 39 & 0.01 & 0.01 & 0.01 & 0.01 & 2 & 2 & [2,3] & [2,5] \\
polblogs & 1224 & 16715 & 351 & 27.3 & 36 & 351 & 600.86 & 87.99 & 92.62 & 163.87 & [5,37] & [5,36] & [5,37] & [6,73] \\
polbooks & 105 & 441 & 25 & 8.4 & 6 & 25 & 0.05 & 0.01 & 0.01 & 0.01 & 4 & 5 & [6,7] & [6,13] \\
pollination-carlinville & 1500 & 15255 & 157 & 20.3 & 18 & 157 & 600.04 & 600.37 & 600.30 & 68.80 & [5,19] & [6,18] & [6,19] & [18,37] \\
pollination-daphni & 797 & 2933 & 124 & 7.4 & 9 & 124 & 292.96 & 0.46 & 0.46 & 0.06 & 5 & 6 & [6,7] & [9,19] \\
pollination-tenerife & 68 & 129 & 17 & 3.8 & 4 & 17 & 0.01 & 0.01 & 0.01 & 0.01 & 3 & 4 & [3,4] & [4,9] \\
ratbrain & 503 & 23030 & 497 & 91.6 & 67 & 497 & 600.50 & 538.67 & 601.38 & 600.83 & [4,68] & [5,67] & [5,68] & [5,135] \\
reactome & 6327 & 147547 & 855 & 46.6 & 191 & 855 & 600.17 & 600.17 & 600.21 & 600.23 & [3,192] & [3,191] & [3,192] & [3,383] \\
residence\_hall & 217 & 1839 & 56 & 16.9 & 11 & 56 & 600.10 & 2.54 & 0.13 & 0.09 & [4,12] & 6 & [10,11] & [11,23] \\
rhesusbrain & 242 & 3054 & 111 & 25.2 & 19 & 111 & 600.09 & 600.19 & 600.23 & 22.19 & [5,20] & [9,19] & [11,20] & [19,39] \\
roget-thesaurus & 1010 & 3648 & 28 & 7.2 & 6 & 28 & 228.14 & 0.13 & 0.03 & 0.02 & 4 & 3 & [6,7] & [6,13] \\
slashdot\_threads & 51083 & 117378 & 2915 & 4.6 & 14 & 2915 & 600.13 & 600.17 & 600.14 & 13.23 & [5,15] & [5,14] & [6,15] & [14,29] \\
soc-Epinions1 & 75879 & 405740 & 3044 & 10.7 & 67 & 3044 & 601.05 & 601.06 & 600.92 & 601.81 & [3,68] & [3,67] & [3,68] & [3,135] \\
soc-Slashdot0811 & 77360 & 469180 & 2539 & 12.1 & 54 & 2539 & 169.08 & 174.42 & 179.09 & 179.02 & [3,55] & [3,54] & [3,55] & [3,109] \\
soc-advogato & 5167 & 39432 & 807 & 15.3 & 25 & 807 & 600.47 & 145.63 & 146.80 & 145.69 & [5,26] & [6,25] & [6,26] & [6,51] \\
soc-gplus & 23628 & 39194 & 2761 & 3.3 & 12 & 2761 & 600.04 & 4.74 & 25.84 & 0.28 & [6,13] & 9 & [7,8] & [12,25] \\
soc-hamsterster & 2426 & 16630 & 273 & 13.7 & 24 & 273 & 600.10 & 83.26 & 131.04 & 229.61 & [5,25] & [11,24] & [23,25] & [24,49] \\
soc-wiki-Vote & 889 & 2914 & 102 & 6.6 & 9 & 102 & 600.09 & 0.19 & 0.11 & 0.04 & [4,10] & 5 & [7,8] & [9,19] \\
sp\_data\_school\_day\_2 & 238 & 5539 & 88 & 46.5 & 33 & 88 & 600.31 & 159.31 & 159.79 & 165.56 & [5,34] & [6,33] & [7,34] & [7,67] \\
teams & 935591 & 1366466 & 2671 & 2.9 & 9 & 2671 & 600.24 & 256.73 & 275.42 & 19.94 & [6,10] & 6 & [6,7] & [9,19] \\
train\_bombing & 64 & 243 & 29 & 7.6 & 10 & 29 & 0.40 & 0.09 & 0.01 & 0.01 & 3 & 5 & 11 & [10,21] \\
tv\_tropes & 152093 & 3232134 & 12400 & 42.5 & 115 & 12400 & 476.72 & 489.41 & 494.33 & 496.33 & [2,116] & [2,115] & [2,116] & [2,231] \\
twittercrawl & 3656 & 154824 & 1084 & 84.7 & 143 & 1084 & 600.46 & 600.68 & 600.37 & 600.47 & [3,144] & [3,143] & [3,144] & [3,287] \\
unicode\_languages & 868 & 1255 & 141 & 2.9 & 4 & 141 & 1.06 & 0.94 & 0.01 & 0.02 & 3 & 4 & [3,4] & [4,9] \\
wafa-ceos & 26 & 93 & 22 & 7.2 & 5 & 22 & 0.01 & 0.01 & 0.01 & 0.01 & 3 & 3 & [5,6] & [5,11] \\
wafa-eies & 45 & 652 & 44 & 29.0 & 24 & 44 & 50.46 & 600.54 & 600.58 & 600.46 & 4 & [14,24] & [22,25] & [13,49] \\
wafa-hightech & 21 & 159 & 20 & 15.1 & 12 & 20 & 0.17 & 0.40 & 0.13 & 10.16 & 3 & 7 & [10,11] & [8,17] \\
wafa-padgett & 15 & 27 & 8 & 3.6 & 3 & 8 & 0.01 & 0.01 & 0.01 & 0.01 & 2 & 2 & [3,4] & [3,7] \\
web-EPA & 4271 & 8909 & 175 & 4.2 & 6 & 175 & 600.11 & 0.10 & 0.21 & 0.02 & [3,7] & 5 & [4,5] & [6,13] \\
web-california & 6175 & 15969 & 199 & 5.2 & 11 & 199 & 600.07 & 0.33 & 3.95 & 0.16 & [4,12] & 11 & [10,11] & [11,23] \\
web-google & 1299 & 2773 & 59 & 4.3 & 17 & 59 & 600.02 & 50.40 & 0.58 & 0.71 & [3,18] & 9 & 18 & [17,35] \\
wikipedia-norm & 1881 & 15372 & 455 & 16.3 & 22 & 455 & 600.14 & 482.15 & 161.62 & 140.63 & [6,23] & [10,22] & [11,23] & [11,45] \\
win95pts & 99 & 112 & 9 & 2.3 & 2 & 9 & 0.01 & 0.01 & 0.01 & 0.01 & 2 & 2 & 3 & [2,5] \\
windsurfers & 43 & 336 & 31 & 15.6 & 11 & 31 & 0.77 & 0.39 & 0.10 & 0.03 & 4 & 6 & [9,10] & [11,23] \\
word\_adjacencies & 112 & 425 & 49 & 7.6 & 6 & 49 & 0.01 & 0.03 & 0.01 & 0.01 & 4 & 4 & [5,6] & [6,13] \\
zewail & 6651 & 54182 & 331 & 16.3 & 18 & 331 & 600.09 & 601.22 & 601.20 & 96.19 & [5,19] & [9,18] & [10,19] & [18,37] \\
\end{longtable}
\endgroup  \end{landscape}

\end{document}